\documentclass[11pt]{article}
\usepackage[utf8]{inputenc}
\usepackage{fullpage}
\usepackage{amsmath, mathtools}
\usepackage{amsfonts} 
\usepackage[english]{babel}
\usepackage{color}
\usepackage[margin=1in]{geometry}
\usepackage{amsthm}
\usepackage{bm}
\usepackage{datetime}
\usepackage{enumerate}

\usepackage{amssymb}
\usepackage[compact]{titlesec}
\usepackage{nicefrac}
\usepackage{mathdots,url}
\usepackage{ifpdf,color}
\usepackage{graphicx}
\usepackage{subfigure}

\usepackage[noend]{algpseudocode}
\usepackage{algorithm}

\newcommand{\K}{\ensuremath{\xi}}
\newcommand{\rh}{\ensuremath{\rho}}

\usepackage{thmtools, thm-restate}

\usepackage{dsfont} 
\newtheorem{theorem}{Theorem}[section]  
\newtheorem{claim}[theorem]{Claim}

\newtheorem{lemma}[theorem]{Lemma}

\theoremstyle{definition}

\newtheorem{definition}[theorem]{Definition}

\newcommand{\R}{\ensuremath{\mathbb{R}}}
\newcommand{\Q}{\ensuremath{\mathbb{Q}}}
\newcommand{\Z}{\ensuremath{\mathbb{Z}}}

\newcommand{\nn}{\ge0}
\newcommand{\po}{>0}

\newcommand{\balval}{14n^2}

\newcommand{\desc}{\operatorname{desc}}
\newcommand{\lca}{\operatorname{lca}}
\newcommand{\parent}{{p}}
\newcommand{\children}{\operatorname{children}}

\newcommand{\LB}{\ensuremath{L}} 
\newcommand{\UB}{\ensuremath{U}} 
\newcommand{\CV}{\operatorname{CV}}
\newcommand{\GIA}{\operatorname{HIA}}
\newcommand{\pred}{\operatorname{pred}}
\newcommand{\Upper}{\operatorname{Upper}}
\newcommand{\Lower}{\operatorname{Lower}}
\newcommand{\bal}{\beta}
\newcommand{\CI}{\operatorname{CI}}
\newcommand{\U}{\operatorname{\Gamma}}
\newcommand{\bck}{\operatorname{Bucket}}

\usepackage[disable,colorinlistoftodos]{todonotes}
\usepackage[colorlinks=true, allcolors=blue]{hyperref}
\newif\ifnotes\notesfalse

\ifnotes
\newcommand{\notename}[2]{{\textcolor{red}{\footnotesize{\bf (#1:} {#2}{\bf ) }}}}

\newcommand{\jnote}[1]{{\notename{Jim}{#1}}}
\newcommand{\lnote}[1]{{\notename{Laci}{#1}}}
\renewcommand{\b}[1]{{\color{blue} #1}}

\else
\newcommand{\notename}[2]{{}}

\newcommand{\jnote}[1]{}
\newcommand{\lnote}[1]{}
\renewcommand{\b}[1]{{ #1}}
\fi

\title{Directed Shortest Paths via Approximate Cost Balancing}

\author{James B. Orlin\thanks{Supported by the ONR Grant N00014--17--1--2194.}\\
Sloan School of Management\\
Massachusetts Institute of Technology\\
 \texttt{jorlin@mit.edu}
\and
L\' aszl\' o A. V\' egh\thanks{Supported by the European Research Council (ERC) under the European Union's Horizon 2020 research and innovation programme (grant agreement ScaleOpt--757481).}\\ 
Department of Mathematics\\
London School of
    Economics and Political Science   \\
\texttt{l.vegh@lse.ac.uk}
}
\date{}

\begin{document}

\maketitle

\begin{abstract}
We present an $O(nm)$ algorithm for all-pairs shortest paths computations in a directed graph with $n$ nodes, $m$ arcs, and nonnegative integer arc costs. This matches the complexity bound attained by Thorup \cite{Thorup1999} for the all-pairs problems in undirected graphs. The main insight is that shortest paths problems with approximately balanced directed cost functions can be solved similarly to the undirected case. The algorithm finds an approximately balanced reduced cost function in an $O(m\sqrt{n}\log n)$ preprocessing step. Using these reduced costs, every shortest path query can be solved in $O(m)$ time using an adaptation of Thorup's component hierarchy method.
The balancing result can also be applied to the $\ell_\infty$-matrix balancing problem.
\end{abstract}


\section{Introduction}

Let $G=(N,A,c)$ be a directed graph with nonnegative arc costs, and $n=|N|$, $m=|A|$.
In this paper, we consider the \emph{single-source shortest paths (SSSP)} and the \emph{all-pairs shortest paths (APSP)} problems. In the SSSP problem, the goal is to find the shortest paths from a given source node $s\in N$ to every other node; in the APSP problem, the goal is to determine the shortest path distances between every pair of nodes.

The seminal approach for SSSP is Dijkstra's 1959 algorithm \cite{Dijkstra1959}. 
 An $O(m+n\log n)$ implementation of this algorithm using 
the Fibonacci heap data structure is due to 
 Fredman and Tarjan  \cite{Fredman1987}.   Under the assumption that all of arc lengths are integral,  Thorup \cite{Thorup2004-integer} improved the running time for SSSP to  $O(m+n\log\log n)$.  Thorup's algorithm uses the word RAM model of computation, discussed in Section \ref{sec:comp-mod}.

For the APSP problem, one can obtain $O(mn+n^2 \log\log n)$ by running the SSSP algorithm of \cite{Thorup2004-integer} $n$ times. This has been the best previously known result for directed graphs. The main contribution of this paper is an $O(mn)$ algorithm for APSP in the word RAM model.

A breakthrough result by Thorup \cite{Thorup1999} obtained a linear time SSSP algorithm in the word RAM model for undirected graphs,  implying $O(mn)$ for APSP. 
Our algorithm matches this bound for undirected graphs: it is based on an
 $O(m\sqrt{n}\log n)$ preprocessing algorithm that enables SSSP queries in $O(m)$ time.

\medskip

Thorup \cite{Thorup1999} uses a  \emph{label setting} algorithm that is similar to Dijkstra's algorithm.    Label setting algorithms maintain upper bounds $D(i)$ on the true shortest path distances $d(i)$ from the origin node $s$ to each node $i$,
and add nodes one-by-one to the set of  permanent nodes $S$.  
 At the time a node $i$ is made permanent,  $D(i)=d(i)$ holds---see \cite[Chapter 4]{AMO}.   In Dijkstra's algorithm,  $D(i) \le D(j)$ is true for all  nodes $j\notin S$ in the iteration when node  $i$ is made permanent.
Relaxing this property is a key in further improvements

Let us define the \emph{bottleneck costs} for
 nodes $i,j\in N$ as
\begin{equation}\label{eq:bottleneck-def}
b(i, j):=\min\left\{\max_{e\in P} c(e): P\textrm{ is an $i$--$j$ path in }G\right\}\, .
\end{equation}

Dinitz \cite{Dinic1978} showed that label setting algorithms are guaranteed to find the shortest path distances if the following is true: whenever a node $i$ is made permanent, $D(i) \le D(j) + b(j, i)$ for all $j\notin S$.   If an algorithm satisfies this weaker condition, then at termination it obtains distances satisfying $d(j) \le d(i) +  b(i, j)$ for all $i$ and $j$, which in turn implies the shortest path optimality conditions: $d(j) \le d(i) + c(i, j)$ for all $(i, j)  \in A$---see Lemma~\ref{lem:shortest-correctness}.

Thorup's algorithm as well as the algorithm presented in this paper rely on this weaker guarantee of correctness. Both algorithms accomplish this by creating a \emph{component hierarchy}---see Definition~\ref{dec:comp-hi} for the variant used in this paper.  Thorup developed this tool for SSSP on undirected networks;   the hierarchy framework was subsequently extended to directed graphs in \cite{Hagerup2000, Pettie2002, Pettie2004}.

Our results also rely on the classical observation that shortest path computations are invariant under shifting the costs by a node potential.
For a potential $\pi: N\to \R$, the \emph{reduced cost}  is defined as $c^\pi(u,v):=c(u,v)+\pi(u)-\pi(v)$. Computing shortest paths for $c$ and any reduced cost $c^\pi$ are equivalent: if $P$ is a $u$--$v$ path, then $c^\pi(P)=c(P)+\pi(u)-\pi(v)$. 

We extend the use of reduced costs to the bottleneck costs.
\begin{equation*}
b^{\pi}(i, j):=\min\left\{\max_{e\in P} c^{\pi}(e): P\textrm{ is an $i$--$j$ path in }G\right\}\, .
\end{equation*}
Our preprocessing step  obtains a reduced cost function satisfying the following $\K$-min-balancedness property for a constant $\K>1$.

\begin{definition}\label{def:balanced}
A strongly connected directed graph $G=(N,A,c)$ with nonnegative arc costs $c\in \R_{\nn}^A$  is \emph{$\K$-min-balanced} for some $\K\ge 1$ if for every arc $e\in A$, there exists a directed cycle $C\subseteq A$ with $e\in C$, such that $c(f)\le \K c(e)$ for all $f\in C$.
\end{definition}

The importance of $\K$-min-balancedness in the context of hierarchy-based algorithms arises from the near-symmetry of the bottleneck values $b(i,j)$.
Lemma~\ref{lem:bottleneck} below shows a graph is $\K$-min-balanced if and only if $b(j,i)\le \K b(i,j)$ for all $i,j\in N$. 
Thorup's component hierarchy for undirected graphs implicitly relies on the fact that $b(i, j) = b(j, i)$ for all nodes $i$ and $j$. For a $\K$-balanced reduced cost function $c^\pi$, the values $b^\pi(i,j)$ and $b^\pi(j,i)$ are within a  factor $\K$.  We can leverage this proximity to use component hierarchies essentially the same way as for undirected graphs in Thorup's original work \cite{Thorup1999}, and achieve the same $O(m)$ complexity for an SSSP query, after an initial $O(m\sqrt{n}\log n)$ balancing algorithm.

This balancedness notion is closely related to the extensive literature on matrix balancing and gives an improvement for approximate $\ell_\infty$-balancing. We give an overview of the related literature in Section~\ref{sec:balance-overview}.

\bigskip

\subsection{Related work}
\subsubsection{The SSSP and  APSP problems}\label{sec:shortest-overview}
In the context of shortest path problems, the choice of the computational model is of high importance. The main choice is between the comparison-addition model with real costs, and variants of word RAM models with integer costs (see Section~\ref{sec:comp-mod}). In the comparison-addition model, additions and comparisons each take $O(1)$ time, regardless of the quantities involved.  Other  operations are not permitted except in so much as they can be simulated using additions and comparisons.

There is an important difference between these computational models in terms of lower bounds: sorting in the comparison-addition model requires $\Omega(n\log n)$, whereas no superlinear lower bound is known for integer sorting. Since Dijkstra's algorithm makes nodes permanent in a non-decreasing order of the shortest path distance $d(i)$ from $s$, the $O(m+n\log n)$ Fibonacci-heap implementation \cite{Fredman1976} is optimal for Dijkstra's algorithm in the comparison-addition model. Moreover, this is still the best known running time for SSSP in this model.

The best running time for APSP in the comparison-addition model is $O(mn+n^2 \log \log n)$ by Pettie \cite{Pettie2004}. This matches the best previous running time bounds for the integer RAM model, where the same bound was previously attained in \cite{Hagerup2000,Thorup2004-integer}.

Pettie's \cite{Pettie2004} algorithm is based on the hierarchy framework. The same paper gives a lower bound that, at first glance, seems to imply that an $O(m)$ running time for the directed SSSP may not be achievable.  

Let $r$ be the ratio between the largest and the smallest nonzero arc cost.   Pettie argued that if a shortest path algorithm for the directed SSSP is based on the hierarchy framework, then the running time of the algorithm is $\Omega(m+\min \{n \log r, n\log n\})$, even if the hierarchy is  provided beforehand. This follows via an information-theoretic argument that is valid both in the comparison-addition as well as in the word RAM models. It uses the fact that any hierarchy approach must make node $i$ permanent before $j$ whenever $d(j)\ge d(i)+b(i,j)$. However, this interpretation of   hierarchy frameworks for directed networks does not allow for replacing the costs by equivalent reduced costs, even though such a transformation may considerably change the bottleneck values $b(i,j)$.  Therefore, his arguments do not contradict our development of an $O(m)$ time algorithm for the directed SSSP.

\medskip

For undirected graphs,  Pettie and Ramachandran \cite{PR2005} solve APSP in $O(mn\log \alpha(m,n))$ in the comparison-addition model, where $\alpha(m,n)$ is the inverse Ackermann function. After an $O(m+\min \{n \log n, n\log \log r\})$ time preprocessing step, every SSSP problem can be solved in time $O(m\log \alpha(m,n))$.

\medskip

For dense graphs, that is, graphs with $m=\Omega(n^2)$ edges, the classical Floyd-Warshall algorithm \cite{Floyd1962,Warshall1962} yields $O(nm)=O(n^3)$.
The first $o(n^3)$ algorithm was given by Fredman \cite{Fredman1976}, in time $O(n^3/\log^{1/3} n)$. This was followed by a long series of improvements with better logarithmic factors, see references in \cite{Williams2014}. In 2014, Williams \cite{Williams2014} achieved a breakthrough with a randomized algorithm  running in time
$n^3/2^{\Omega(\sqrt{\log n})}$, by speeding up min-plus (tropical) matrix multiplication using tools from circuit complexity. A deterministic algorithm of the same asymptotic running time was obtained by Chan and Williams \cite{Chan2021}.

\subsubsection{Approximate graph and matrix balancing}\label{sec:balance-overview}
Our notion of $\K$-min-balanced graphs is closely related to previous work on graph and matrix balancing. For $\K=1$, we 
 simply say that $G$ is min-balanced. 
A graph $G$ is min-balanced if and only if for each proper subset $S$ of nodes, the following is true:   the minimum cost over arcs entering $S$ is at equal to the minimum cost over arcs leaving $S$---see Lemma \ref{lem:bottleneck}.

Schneider and Schneider \cite{Schneider1991} defined max-balanced graphs where for every subset $S$, the maximum cost over arcs entering  $S$ equals the maximum cost over arcs leaving $S$. For each $e \in E$,  let $c'(e)= c_{\max} - c(e)$.   Then $G$ is max-balanced with respect to $c' $ if and only if $G$ is min-balanced with respect to $c$.
For exact min/max-balancing, the running time $O(mn+n^2\log n)$ by Young, Tarjan, and Orlin \cite{Young1991} is still the best known complexity bound. Relaxing the exactness condition, we give an 
 $O\left( 2^\rho (\rho+1)\sqrt{n} \log n\right)$ algorithm  for $\K$-min-balancing for any  $\rho\in \Z_{\ge0}$, $\K=1+1/2^{\rho-1}$.

\medskip

Min-balancing is a generalization of the min-mean cycle problem: if $C$ is a min-mean cycle, then any min-balanced residual cost function satisfies $c^\pi(e)\ge \mu$ for all $e\in E$ and $c^\pi(e)=\mu$ for $e\in C$ for some $\mu\in \R$. In fact, following \cite{Schneider1991}, one can solve min-balancing as a sequence of min-mean cycle computations; see the discussion after Theorem~\ref{thm:balance-main}. 
Karp's $O(mn)$ algorithm from 1978 \cite{Karp1978} is still the best known strongly polynomial algorithm for min-mean cycle problem.
Weakly polynomial algorithms that run in $O(m\sqrt{n}\log(nC))$ time  were given by Orlin and Ahuja \cite{Orlin1992}  and by    McCormick \cite{McCormick1993}. The latter provides a scaling algorithm based on the same subroutine of Goldberg \cite{Goldberg1995} that plays a key role in our balancing algorithm. The algorithms \cite{McCormick1993,Orlin1992} easily extend to finding an $\varepsilon$-approximate min-mean cycle in 
 $O(m\sqrt{n}\log(n/\varepsilon))$ time.   That is, finding a reduced cost $c^\pi$, a cycle $C$, and a value $\mu$ such that $c^\pi(e)\ge \mu$ for all $e\in E$ and $c^\pi(e)\le (1+\varepsilon)\mu$ for all $e\in C$.

A restricted case of APSP is the problem of finding the shortest cycle in a network.   Orlin and Sede\~no-Noda \cite{orlin2017nm} show how to solve the shortest cycle problem in $O(nm)$ time by solving a sequence of $n$ (truncated) shortest path problems, each in $O(m)$ time.   Their preprocessing algorithm was the solution of a minimum cycle mean problem in $O(nm)$ time.   However---analogously to the approach in this paper---they could have relied instead on \cite{McCormick1993,Orlin1992} to find a 2-approximation of the minimum cycle mean in $O(m\sqrt{n}\log n)$ time.

\medskip

We say that  a graph is \emph{weakly max-balanced} if for every node $v\in N$, the maximum cost over arcs entering $v$ equals the maximum cost over arcs leaving $v$; that is, we require the property in the definition of max-balancing only for singleton sets $S=\{v\}$. 

This notion corresponds to the well-studied \emph{matrix balancing} problem: given a nonnegative matrix $M\in \R^{n\times n}$, and a parameter $p\ge 1$, find a positive diagonal matrix $D$ such that in $DMD^{-1}$, the $p$-norm of the $i$-th column equals the $p$-norm of the $i$-th row. Given $G=(N,A,c)$, we let $M_{ij}=e^{c_{ij}}$ if $(i,j)\in A$ and $M_{ij}=0$ otherwise. Then,  balancing $M$ in $\infty$-norm amounts to finding a weakly max-balanced reduced cost $c^\pi$.

Matrix balancing was introduced by Osborne \cite{Osborne1960} as a preconditioning step for eigenvalue computations. He also proposed a natural iterative algorithm for $\ell_2$-norm balancing. Parlett and Reinsch \cite{Parlett1969} extended this algorithm to other norms. Schulman and Sinclair \cite{Schulman2017} showed that a natural variant of the Osborne--Parlett--Reinsch (OPR) algorithm finds an $\varepsilon$-approximately balanced solution in $\ell_\infty$ norm in time $O(n^3\log (n\rho/\varepsilon))$, where $\rho$ is the initial imbalance. Ostrovsky, Rabani, and Yousefi \cite{Ostrovsky2017} give polynomial bounds for variants of the  OPR algorithm for fixed finite $p$ values, in particular, 
$O(m+n^2\varepsilon^{-2}\log w)$ for a weighted randomized variant, where $w$ is the ratio of the sum of the entries over the minimum nonzero entry, and $m$ is the number of nonzero entries.
Recently, Altschuler and Parillo \cite{Altschuler2022} showed an $\tilde O(m\varepsilon^{-2}\log w )$ bound for a simpler randomized variant of OPR.
Cohen~et al.~\cite{Cohen2017} use second order optimization techniques to attain $\tilde O(m\log \kappa \log^2 (nw/\varepsilon))$, where $\kappa$ is the ratio between the maximum and minimum entries of the optimal rescaling matrix $D$; similar running times follow from \cite{Allen2017}.
The value $\kappa$ may be exponentially large; the paper \cite{Cohen2017} also shows a $\tilde O(m^{1.5} \log (nw/\varepsilon))$ bound via interior point methods using fast Laplacian solvers.\footnote{In the quoted running times, $\tilde O(.)$ hides polylogarithmic factors. Various papers define $\varepsilon$-accuracy in different ways; here, we adapt the statements to $\ell_1$-accuracy as in \cite{Altschuler2022}.}

\medskip

Our graph balancing problem corresponds to $\ell_\infty$ matrix balancing.
 Except for \cite{Schulman2017}, the above works are applicable for finite $\ell_p$ norms \b{only}. Compared to \cite{Schulman2017}, our approximate balancing algorithm has lower polynomial terms, but our running time depends linearly on $1/\varepsilon$ instead of a logarithmic dependence.\footnote{%
We note that, in contrast to the previous work, we consider min- rather than max-balancing. The exact min- and max-balancing problems can be transformed to each other by setting $c'(e)= c_{\max} - c(e)$; however, such a reduction does not preserve multiplicative approximation factors, and hence our result cannot be directly compared with \cite{Schulman2017}. Nevertheless, it seems that both algorithms can be adaptable to both the min and max settings. Such extensions are not included in this paper.}

\subsection{Overview}
The rest of the paper is structured as follows. Section~\ref{sec:prelim} introduces  notation and basic concepts, including  the directed variant of component hierarchies used in this paper, and 
the comparison-addition and word RAM computational models. Section~\ref{sec:balance} is dedicated to the approximate min-balancing algorithm.
 The algorithm is developed in several steps: a key ingredient is a subroutine by Goldberg \cite{Goldberg1995} that easily gives rise to a weakly polynomial algorithm. In order to achieve a strongly polynomial bound, we need a further preprocessing step to achieve an initial `rough balancing'. An additional technical contribution is  a new variant of the Union-Find data structure, called Union-Find-Increase. At the beginning of Section~\ref{sec:balance}, we give a detailed overview of the overall algorithm and the various subsections.

 In Section~\ref{sec:shortest-path}, we describe the shortest path algorithm for 3-min-balanced directed graphs. This is very similar to Thorup's original algorithm \cite{Thorup1999}. However, the setting is different, and we use a slightly different notion of the component hierarchy. For completeness, we include a concise description of the algorithm and the proof of correctness.
Concluding remarks are given in the final Section~\ref{sec:conclusions}.

\section{Notation and preliminaries}\label{sec:prelim}
For an integer $k$, we let $[k]=\{1,2,\ldots,k\}$. We let $\Z_{\nn}$ denote the nonnegative integers and let $\Z_{\po}$ denote the positive integers; similarly for $\Q_{\nn}$, $\Q_{\po}$, $\R_{\nn}$, and $\R_{\po}$. We let $\log x=\log_2 x$ refer to base 2 logarithm unless stated otherwise. For a vector $z\in \R^T$ and $S\subseteq T$, we let $z|_{S}\in\R^S$ denote the restriction of $z$ to $S$.

Throughout, we let $G=(N,A,c)$ be a directed graph with nonnegative  arc costs $c\in \R^A_{\nn}$. Let $m=|A|$ and let  
$n$ denote the smallest integer power of $2$ greater than or equal to $|N|$; we  assume $n,m\ge 2$. This choice instead of $n=|N|$ will be convenient in the word RAM model. 
We  use $C_{\min}=\min_{e\in E}c(e)$ and $C_{\max}=\max_{e\in E}c(e)$ to denote the smallest and largest values of the cost function.
All graphs considered will be simple and loopless.

For a node $i\in N$, we let $A(i)$ denote the set of the outgoing arcs from $i$.
For an arc set $F\subseteq A$, we let $N(F)$ denote the set of nodes incident to $F$.  For a node set $X\subseteq N$, let $A[X]$ denote the set of arcs in $A$ with both endpoints inside $X$.

For a node set $S\subseteq N$,  let $\bar S= N\setminus S$ denote the complement of $S$.  We let $(S, \bar S)\subseteq A$ denote the set of arcs directed from a node in $S$ to a node in $\bar S$.

 For a node set $Z\subseteq N$, we  denote the graph obtained by contracting $Z$ by $G/Z=(N',A',c')$. Here, $N'=(N\setminus Z)\cup\{z\}$;  $z$ represents the contracted node set. We include every arc $(i,j)\in A$ in $A'$ with the same cost if $i,j\notin Z$. Arcs with both endpoints in $Z$ are deleted. If $i\in Z$ or $j\in Z$, the corresponding endpoint is replaced by $z$. In case parallel arcs are created, we only keep one with the smallest cost.
For a partition $\mathcal{P}=(P_1,P_2,\ldots,P_k)$ of $N$, the contraction $G/\mathcal{P}$ denotes the graph obtained after contracting (in an arbitrary order) each of the sets $P_i$, $i\in k$ in $G$.

We will assume that  $G=(N,A,c)$ is \emph{strongly connected}; that is, a directed path exists between any two nodes. If the input is not strongly connected, then we preprocess the graph as follows. We find  the strongly connected components in $O(n+m)$ time using Tarjan's algorithm \cite{Tarjan1972}. We select a value $M$ greater than the sum of all arc costs, pick one node in each strongly connected component,  add a directed cycle on these nodes, and set the cost of these arcs to $M$. This results in a strongly connected graph $G'=(N,A',c')$ with $|A'|=O(m+n)$. Computing shortest paths in $G'$ provides the shortest paths in $G$; if the shortest path distance between nodes $i$ and $j$ in $G'$ is at least $M$, then $j$ is not reachable from $i$ in $G$.

\paragraph{Dijkstra's algorithm}
Dijkstra's  algorithm \cite{Dijkstra1959} is the starting point of the fastest algorithms for  SSSP and APSP. We now give a brief overview of the key steps.
The algorithm maintains distance labels $D(i)$ for each node $i$ that are upper bounds on $d(i)$, the shortest path distance from $s$.  The algorithm adds nodes one-by-one to a \emph{permanent node set} $S$ with the property that $D(i)=d(i)$ for every $i\in S$.
Further, for every $i\in N\setminus S$, $D(i)$ is the length of a shortest $s$--$i$ path in the subgraph induced by the node set $S\cup\{i\}$.

These are initialized as $D(s)=0$, $D(i)=\infty$ for $i\in N\setminus \{i\}$, and $S=\emptyset$.
  Every iteration adds a new node to $S$, selecting the node  $i\in N\setminus S$ with the smallest label $D(i)$. Then, the outgoing arcs $(i,j)$ are considered, and  $D(j)$ is updated to $\min\{D(j),D(i)+c(i,j)\}$. The crucial property of the analysis is that this selection rule is correct, that is, for $i\in\arg\min\{D(j): j\in N\setminus S\}$, we must have $D(i)=d(i)$.

\paragraph{Bottleneck costs in balanced graphs}
Our shortest path algorithm requires the input graph to be $3$-min-balanced---see Definition~\ref{def:balanced}. As shown next, the bottleneck costs are approximately balanced in such graphs.

Recall the definition of the bottleneck cost $b(i,j)$ in \eqref{eq:bottleneck-def}.   We extend the definition to non-empty disjoint subsets $S,T \subsetneq N$ as   $b(S, T) := \min \{ b(i, j) : i \in S,\, j \in T \}$.   Equivalently, 
$b(S, \bar S) = \min \{c(i, j) : i \in S, j \in \bar S \}$.
 By a \emph{bottleneck} $i$--$j$ path we mean an $i$--$j$ path where the maximum arc cost is $b(i,j)$.

\begin{lemma}\label{lem:bottleneck}
The following are equivalent.
\begin{enumerate}[(1)]
\item  $G$ is $\K$-min-balanced.
\item  For all proper subsets $\emptyset\neq S \subsetneq N$, $b(\bar S, S) \le \K b(S, \bar S)$.
\item  For all $i \in N$ and $j \in N$, $b(j, i) \le \K b(i, j)$.
\end{enumerate}

\end{lemma}

\begin{proof}
\emph{(1) $\Rightarrow$ (2).}    Suppose that $G$ is $\K$-min-balanced and $\emptyset\neq S \subsetneq  N$.  
Choose $e \in\arg\min\{c(e) : e \in (S, \bar S)\}$;    thus, $c(e) = b(S, \bar S)$.   Let $C$ be the bottleneck cycle containing $e$.   Because $C$ contains an arc $f$ of $(\bar S, S)$, the following is true:   
$b(\bar S, S) \le c(f) \le \K c(e) =\K b(S, \bar S)$.  

\paragraph{(2) $\Rightarrow$ (3).}   
Suppose that (2) is true.  For given nodes $i$ and $j$, let  $S = \{k \in N : b(j, k) \le  \K b(i, j)\}$.   Clearly, $j \in S$.  We show by contradiction  that $i \in S$, and consequently, $b(j, i) \le \K b(i, j)$.    Suppose that $i \in\bar S$.   Let $e \in \arg\min\{c(e) : e \in (S, \bar S)\}$, and suppose that $e = (h, \ell)$.  
Then $b(j, h) \le  \K b(i, j)$ because $h \in S$; further,
$c(h, \ell) = b(S, \bar S) \le \K b(\bar S, S) \le \K b(i, j)$ by (2) and the fact that the bottleneck path from $i$ to $j$ includes an arc of $(\bar S, S)$.    Then 
$b(j, \ell) \le \max \{ b(j, h), c(h, \ell) \} \le \K b(i, j)$.  But this implies that $\ell \in S$, a contradiction.

\paragraph{(3) $\Rightarrow$ (1).}   Suppose that (3) is true.  
Let $e=(j,i)$ be any arc of $A$;   note that $b(j, i) \le c(e)$.  Let $P$ be a path from $i$ to $j$ with arcs of length at most $b(i, j)$, and let $C = P \cup \{e\}$.  Then $C$ is a cycle, and  
$\max\{c(f): f \in C\} \le \max\{b(i, j), c(e)\} \le \max\{\K b(j, i), c(e)\} \le \K c(e)$.   Thus, $G$ is $\K$-min-balanced.
\end{proof}

\paragraph{The component hierarchy} 
We now introduce the concept of a component hierarchy. This is a variant of Thorup's \cite{Thorup1999} component hierarchy,  adapted for approximately min-balanced directed graphs. The papers \cite{Hagerup2000,Pettie2002,Pettie2004} also use component hierarchies for directed graphs. However, our notion exploits the $\K$-min-balanced property, and will be more similar to the undirected concept \cite{Thorup1999} in that it does not impose orderings of the children of the vertices.

We use the standard terminology for a tree $(V',E')$ rooted at $r\in V'$. 
\begin{itemize}
	\item For $v\in V'\setminus\{r\}$, the \emph{parent} $\parent(v)$ of $v$ is the first vertex after $v$ on the unique path in the tree from $v$ to $r$. All nodes in the path are called the \emph{ancestors} of $v$.
	\item For $v\in V'$, $\children(v)\subseteq V'$ is the set of nodes $u$ such that $\parent(u)=v$.
    \item For every $v\in V'$,  $\desc(v)\subseteq V'$ is the set of nodes in the subtree rooted at  $v$. 
	\item For $u,v\in V'$, $\lca(u,v)$ is the \emph{least common ancestor} of $u$ and $v$, i.e., the unique vertex on the $u$--$v$ path in $E'$ that is an ancestor of both $u$ and $v$.
\end{itemize}
\begin{definition}\label{dec:comp-hi}
 The tuple $(V\cup N,E,r,a)$ is called a \emph{component hierarchy of $G$}  for a strongly connected directed graph $G=(N,A,c)$ if
\begin{itemize}
\item $(V\cup N,E)$ is a tree with root $r\in V$, and $N$ is the set of leaves. 
\item The vector $a:V \to \Z_{\po}$ is such that each $a(v)$ is an integer power of 2. For every $v\in V\setminus\{r\}$, $a(v)\le a(p(v))/2$.
\item  For any $i,j\in N$ with $\lca(i,j)=v$, we have $a(v)\le b(i,j)\le 3 a(v)$; moreover, there exists an $i$--$j$ path $P$ inside $\desc(v)\cap N$ such that $c(e)\le 3 a(v)$ for every arc $e\in P$.
\end{itemize}
\end{definition}

\subsection{Computational models}\label{sec:comp-mod}
Our results use two different computational models. The approximate min-balancing algorithm in Section~\ref{sec:balance} can be implemented in the 
more restrictive comparison-addition model.
However, the word RAM model is needed for constructing the component hierarchy: we require the operation of rounding down numbers to the nearest power of two in order to obtain  $a(v)$ values that are powers of two.
The shortest path algorithm in Section~\ref{sec:shortest-path} uses integer arithmetic  in two parts: {\em(1)} storing vertices and nodes in buckets, and {\em(2)} in the Split/FindMin data structure.

\paragraph{The comparison-addition model}  The input is a set of real numbers, and only addition and comparison operations are allowed; each takes constant time. Subtraction can be easily simulated with a constant overhead  by representing numbers in the form $\alpha-\beta$. Multiplication by an integer $N$ can be simulated by $O(\log N)$ additions. See \cite{Pettie2002,PR2005} for more details.

The algorithms in Section~\ref{sec:balance} also include division by a power of 2 in a restricted sense: for a value $\gamma= O(\log n)$, we require that all calculated values can be expressed as sums of quotients in the form  $w=\sum_{i=0}^\gamma w_i/2^i$, where the $w_i$ values can be obtained as a difference of two sums of input values.   We  can work with such numbers in the comparison-addition model  by representing such a sum by ordered pairs $(w_1, d_1),\ldots,(w_i,d_i)$.   

One can convert the sum of quotients into a single quotient with denominator
$\le 2^\gamma$ in time $O(\gamma)$.   Moreover, additions, subtractions, and comparisons of sums of quotients can each be carried out in $O(\gamma+1)$ steps.   
When the approximate balancing algorithms are used as preprocessing for the APSP, one can multiply the outputs by $2^\gamma$ at termination.  Then the shortest path algorithms will determine the shortest path trees for the original problem as well.

\paragraph{The word RAM model} 
We use the standard random access machine model, where every memory cell can store an integer of $w$ bits. For convenience, we assume the word size is at least
$\log (nC_{\max})$ so that  each input and output number fits into a single word. 

There is no universally accepted computational model for integer weights. We use the same model as in \cite{Hagerup2000};  this is more restrictive than the one in \cite{Thorup1999}, which also allows arbitrary multiplications. In our model,
unit-time operations include comparison, addition, subtraction, bit shifts by an arbitrary number of positions, and bitwise boolean operations.

We do not allow multiplications and divisions in general.  However, the bit shift operations  enable multiplications  by integer powers of 2 in $O(1)$ time. Due to the assumption that $n$ is a power of 2, multiplying
by a monomial term such as $bn^k$ can be done in $O(1)$ time if $b,k=O(1)$.
We will use  divisions by powers of 2. These operations can be simulated with constant overhead, maintaining a representation $a/2^b$ of the occurring numbers. Throughout, we maintain numbers in such representation with $b=O(\log n)$.

We highlight the only two operations involving integer arithmetics that are used for constructing the component hierarchy in Section~\ref{sec:balance-strong} and for the bucketing operations in Section~\ref{sec:shortest-describe}.
\begin{enumerate}[(i)]
    \item Given $r\in \Z_{\nn}$, compute the largest integer power of $2$ smaller or equal than $r$; we denote this by $\lfloor r\rfloor_{2}$.
    \item   Given $r\in \Z_{\nn}$, and $b\in \Z_{\nn}$,    compute $\lfloor r/2^b \rfloor$.
\end{enumerate}
All other integer operations are only needed for Split/FindMin; we discuss this in more detail in Section~\ref{sec:split-findmin}.

\medskip

Any running time bound obtained in the comparison-addition model is directly applicable to the word RAM model. Bounds in the comparison-addition model can be worse than in the word RAM model. In particular, restricted divisions in the comparison-addition model require $O(\gamma)$ time for numbers in the sum of quotients form $w=\sum_{i=0}^\gamma w_i/2^i$, in contrast to $O(1)$  in the word RAM model where bit shift operations are permitted.


\section{An algorithm for approximate min-balancing}\label{sec:balance}
This section is dedicated to the proof of the following theorem.
 The algorithm asserted in the theorem is  Algorithm~\ref{alg:balance} in Section~\ref{sec:balance-strong}.

\begin{theorem}\label{thm:balance-main}
Assume we are given a strongly connected directed graph
$G=(N,A,c)$  with  arc costs $c\in \R_{\nn}^A$, and a parameter $\rh\in\Z_{\nn}$; let
$\K:=1+1/2^{\rh-1}$.
\begin{enumerate}[(a)]
\item \label{part:min-bal-alg}
There exists an  $O\left(2^\rh (\rh+1)\cdot m\sqrt{n}\log n\right)$ time algorithm in the comparison-addition model that finds a potential $\pi\in \R^N$ such that $c^\pi$ is $\K$-min-balanced. 
\item\label{part:comp-hi}
For $\rh=0$ and $\K=3$ and an integer input $c\in \Z_{\nn}^A$, we can obtain a potential  $\pi\in \Q^N$ and
a component hierarchy of $(N,A,c^\pi)$ in time $O\left(m\sqrt{n}\log n\right)$ in the word RAM model. Further, all $\pi(v)$ values  are integer multiples of $1/(4n^3)$. 
\end{enumerate}
\end{theorem}

It is instructive to start the overview from exact min-balancing, that is, $\K=1$, even though our algorithm is not applicable to this case. For $\K=1$,  the exact max-balancing algorithms \cite{Schneider1991,Young1991} can be used (by negating the costs).  A simple and natural algorithm (see \cite{Schneider1991}) is based on the iterative application of min-mean cycle finding. First, find all arcs that are in a min-mean cycle in the graph; let $\mu\ge 0$ denote the minimum cycle mean value, and $F$ the set of all arcs in such cycles. Every arc $e\in F$ must have $c^\pi(e)=\mu$ if $c^\pi$ is a min-balanced reduced cost function.

 It is easy to see that the min-cycle mean  algorithm produces a potential $\pi$ such that $c^\pi(e)\ge \mu$ for all $e\in E$, and $c^\pi(e)=\mu$ for all $e\in F$. We can then contract all strongly connected components of $F$, and recurse on the contracted graph, by repeatedly modifying the potential $\pi$ and contracting the components of min-mean cycles.

The current best  running times are  $O(mn+n^2\log n)$  for min-balancing \cite{Young1991} and $O(mn)$ for minimum-mean cycle computation \cite{Karp1978}.  Both  these running times are substantially higher than the overall running time in Theorem~\ref{thm:balance-main}.

We can thus only afford to approximately compute min-mean cycles. This can be achieved faster using a subroutine in Goldberg's paper \cite{Goldberg1995}, originally developed for a weakly polynomial algorithm for negative cycle detection. 
There are some technical differences from \cite{Goldberg1995}; we present the detailed description of the subroutine and the proof of correctness in Section~\ref{sec:goldberg}. 

The input to the subroutine \textsc{Small-Cycles} is a strongly connected directed graph with  minimum  arc cost $L$ and a parameter $D>0$. In time $O( m\sqrt{n})$, it finds a reduced cost $c^\pi$ and  strongly connected components of arcs with reduced cost in the range of $[L,L+2D]$. At the same time, the reduced cost of every arc between different components is at least $L+D$.

If the input graph has positive arc costs, the iterative application of this subroutine yields a simple weakly polynomial algorithm with running time $O\left(2^\rh(\rh+1)\cdot {m\sqrt{n} \log\left (nC_{\max}/C_{\min}\right)}\right)$,
as described in Section~\ref{sec:simple-weakly}. 

In order to turn this into a strongly polynomial algorithm, in Section~\ref{sec:rough-balance} we start by a preprocessing algorithm that finds a $\balval$-min-balanced reduced cost. An important step in this algorithm is to  determine the \emph{balance values} $\bal(e)$ for all arcs $e\in E$; this is defined as the smallest value $b$ such that $G$ contains a cycle $C$ with $e\in C$ and $c(f)\le b$ for all $f\in C$. These balance values can be efficiently found using a simple recursive framework. 

The strongly polynomial algorithm in Section~\ref{sec:balance-strong} requires the input cost function to be  $\balval$-min-balanced. How can we benefit from this `rough' balance of the input? The weakly polynomial algorithm consists of 
$O\left(2^\rh(\rh+1)\cdot{\log(nC_{\max}/C_{\min})}\right)$ calls to \textsc{Small-Cycles}. 
If each call  uses the entire arc set in the current contracted graph, we obtain a total running time $O\left(2^\rh(\rh+1)\cdot {m\sqrt{n} \log\left (nC_{\max}/C_{\min}\right)}\right)$ as above.
 However, when running \textsc{Small-Cycles} with parameter $L$, it is possible to restrict attention to arcs $e$ with $c(e) \le 2nL$.  We refer to such arcs as \emph{active}.  If the input is assumed to be a $\balval$-min-balanced cost-function, then each arc is active for $O\left(2^\rho \log n\right)$ calls of \textsc{Small-Cycles} prior to being contracted.   Thus each arc contributes $O\left(2^\rho(\rho+1) \sqrt{n}\log n\right)$
to the total running time.  

In the weakly polynomial algorithm, the parameter $L$ giving a lower bound on the minimum reduced cost of non-contracted arcs increases  by a factor at most $1+1/2^{\rho}$ in each iteration. To avoid the dependence on $C_{\max}/C_{\min}$ in the strongly polynomial algorithm, this value may sometimes `jump' by large amounts in iterations with no active arcs.

 An important technical detail is the maintenance of the reduced costs. In every iteration, we only directly maintain $c^\pi(e)$ for the active arcs. Querying the reduced cost of a newly activated arc is nontrivial, since one or both of its endpoints may have been part of one or more contracted cycles, each of which corresponds to a node in the contracted graph.   To compute the potential of an original node $i$, we need to add to the potential of node $i$ the potentials of every contracted node $j$ that contains node $i$.  
%
 We develop a new extension of the Union-Find data structure, called Union-Find-Increase by incorporating a new `increase' operation. This is described in Section~\ref{sec:update-costs}.

\paragraph{Contractions and preprocessing}
We  use contractions several times. Whenever a set $S$ is contracted, we let $s$ be the contracted node, and  set the potential  $\pi_{s}=0$. For each arc with one endpoint in $S$, we keep the same reduced cost as immediately before the contraction.

On multiple occasions we need the subroutine
\textsc{Strongly-connected}$(N,A)$ that implements  Tarjan's algorithm \cite{Tarjan1972} to find the strongly connected components of the directed graph $(N,A)$ in time $O(|N|+|A|)$. The output includes the strongly connected components $(N_1,A_1)$, $(N_2,A_2),\ldots, (N_k,A_k)$ in the topological order, namely, for every arc $(u,v)\in A$ such that $u\in N_i$, $v\in N_j$, it must hold that $i\le j$. 

\medskip

In Theorem~\ref{thm:balance-main}, the input is a nonnegative  cost function. For our algorithm, it is more convenient to assume a strictly positive cost function. We now show how the nonnegative case can be reduced to the strictly positive case by a simple $O(m)$ time preprocessing. Recall the notation 
 $C_{\min}=\min_{e\in E}c(e)$ and $C_{\max}=\max_{e\in E}c(e)$.

We first call \textsc{Strongly-connected}$(N,A_0)$ on the subgraph of $0$-cost arcs $A_0$. We contract all strongly connected components, and keep the notation $G=(N,A)$ for the contracted graph, where the output of the subroutine gives a topological ordering  $N=\{v_1,v_2,\ldots,v_n\}$ such that for every $0$-cost arc $(v_i,v_j)$, we must have $i<j$. We let $C'$ denote the smallest nonzero arc cost, and set $\pi({v_i})=-iC'/n$. Then, it is easy to see that  $c^\pi(e)\ge C'/n$ for every $e\in A$. 

We then replace the cost function $c$ by $nc^\pi$, after which we obtain a cost function $c'$ with $C'\le c'(e)\le n(C'+C_{\max})$ for every $e\in A$.
This preprocessing algorithm can be implemented in $O(m\log n)$ time in the comparison-addition model.

\subsection{A simple weakly polynomial variant}\label{sec:simple-weakly}
The following subroutine is a variant of  \textsc{Refine} in  Goldberg's paper \cite{Goldberg1995}.

\begin{algorithm}
    \caption{\textsc{Small-cycles}}
    \begin{algorithmic}[1]
        \Require{A directed graph 
$G=(N,A,c)$ with a cost function $c \in \R^A$, and $L\in \R$, $D\in\R_{\po}$  such that $c(e)\ge L$ for all $e\in A$.}
                \Ensure{A partition $\mathcal{P}=(P_1,P_2,\ldots,P_k)$  of the node set $N$ and a potential vector $\pi\in  \R^N$ such that 
\begin{enumerate}[(i)]
    \item For every $i\in [k]$, $c^\pi(e)\ge L$ for every $e\in A[P_i]$, and
     $P_i$ is strongly connected in the subgraph of arcs $\{e\in A[P_i]:\, L\le c^\pi(e)\le L+2D\}$; 
    \item $c^\pi(e)\ge L+D$ for all $e\in A\setminus\left(\cup_{i\in [k]} A[P_i]\right)$;
    \item $-|N| D\le \pi(v)\le 0$, and $\pi(v)$ is an integer multiple of $D$ for all $v\in N$.
\end{enumerate}}
    \end{algorithmic}
\end{algorithm}

\begin{lemma}\label{lem:small-cycles}
The subroutine \textsc{Small-cycles}$(L,D,N,A,c)$ can be implemented in $O(|A|\sqrt{|N|}\log |N|)$ time in the comparison-addition model. 
\end{lemma}
The proof adapts the argument in \cite{Goldberg1995}; it is deferred to Section~\ref{sec:goldberg}.
We now summarize the weakly polynomial algorithm \textsc{Simple-Min-Balance} (Algorithm~\ref{alg:weakly}). 
We initialize $L_1=C_{\min}$ and $D_1=L_1/2^\rh$. Every iteration calls  \textsc{Small-Cycles} for the current values of $L_t$ and $D_t$. In Step~\ref{line:contract-w}, we contract each subset (some or all of which may be singletons) in the partition $\mathcal{P}_t$ returned by the subroutine, and iterate with the returned reduced cost, setting the new value  $L_{t+1}=L_t+D_t$. We update $D_{t+1}$ to $L_{t+1}/2^\rho$ whenever $t$ is an integer multiple of $2^\rho$; otherwise, we keep $D_{t+1}=D_t$. Thus, the value of $L_t$ doubles in every $2^\rho$ iterations.

\begin{algorithm}
    \caption{\textsc{Simple-min-balance}}\label{alg:weakly}
    \begin{algorithmic}[1]
        \Require{A strongly connected directed graph $G=(N,A,c)$  with  $c\in\R_{\po}^A$, parameters $\rh\in\Z_{\nn}$ and $\K=1+1/2^{\rh-1}$.}
                \Ensure{A potential $\pi\in \R^N$ such that $c^\pi$ is  $\K$-min-balanced.}
                \State $(\hat N_1,\hat A_1,\hat c_1)\gets (N,A,c)$ ; $t\gets 1$ ;
                \State $L_1\gets \min_{e\in A} c(e)$ ; $D_1\gets L_1/2^\rho$ 
                \While{$|\hat N_t|>1$}
                    \State $(\mathcal{P}_t,p_t)\gets $ \Call{Small-Cycles}{$L_t,D_t,\hat N_t,\hat A_t,\hat c_t$} ;
                    \State  $(\hat N_{t+1},\hat A_{t+1},\hat c_{t+1})\gets (\hat N_t,\hat A_t,\hat c_t^{p_t})/\mathcal{P}_t $ ;\label{line:contract-w}
               \State $L_{t+1}\gets L_t+D_t$ ;
                   \If{$t$ is an integer multiple of $2^\rh$} 
                      $D_{t+1}\gets L_{t+1}/2^\rh$ ; 
                    \Else
                      \ $D_{t+1}\gets D_t$ ;\EndIf 
                   \State $t\gets t+1$ ; 
                \EndWhile
                \State Uncontract $(\hat N_t,\hat A_t,\hat c_t)$, and compute the overall potential $\pi\in\R^N$ ; \label{line:uncontract-w}
    \State \Return{$\pi$.} 
    \end{algorithmic}
\end{algorithm}

We let $(\hat N_t, \hat A_t)$ denote the contracted graph at iteration $t$.
The algorithm terminates when  $\hat N_t$ has a single node only, at iteration $t=T$. 

\paragraph{Uncontraction} In  the final step of the algorithm,
we uncontract all sets in  the reverse order of contractions. We start by setting $\pi=p_T$.
Assume a set $S$ was contracted to a node $s$ in iteration $t$, and we have uncontracted all sets from iterations $t+1,\ldots,T$. When uncontracting $S$, for every $v\in S$ we set $\pi(v)=p_t(v)+\pi(s)$, i.e., the potential right before contraction, plus the potential of $s$ accumulated during the uncontraction steps.
 This takes time $O(n')$ where $n'$ is the total size of all sets contracted during the algorithm; it is easy to bound $n'\le 2n$. Thus, the total time for uncontraction is $O(n)$.

\begin{lemma}\label{lem:weakly-analysis}
Algorithm~\ref{alg:weakly} finds a $\K$-min-balanced cost function in time
$O\left(2^\rh(\rh+1)\cdot m\sqrt{n} \log (nC_{\max}/C_{\min})\right)$ in the comparison-addition model.
\end{lemma}
\begin{proof}
At initialization, $L_1=C_{\min}$, and $L_t$ increases by a factor $2$ in every $2^\rh$ iterations. At every iteration, we can extend the cost function $\hat c_t$ to the original arc set $A$: for an arc $e$ contracted in an earlier iteration $\tau<t$, we let $\hat c_t(e)=\hat c_{\tau}(e)$ represent the value right before the contraction. It is easy to see that this extension of $\hat c_t$ to $A$ gives a valid reduced cost of $c$. 

Throughout, we have that $L_t\le \hat c_t(e)$ for all $e\in \hat A_t$, and $\hat c_t(e)\ge 0$ for all contracted arcs. Thus, for any cycle $C\subseteq A$ that contains some non-contracted arcs in $\hat A_t$, $2L_t\le \hat c_t(C)=c(C)\le nC_{\max}$ holds. Consequently,  $L_t\le nC_{\max}/2$ throughout, implying 
a bound $O(2^\rho \log (nC_{\max}/C_{\min}))$ on the number of iterations.

As explained above, the final uncontraction and computing $\pi$ can be implemented in $O(n)$ time.
To show that the final $c^\pi$ is $\K$-min-balanced, consider an arc $e\in A$, and assume it was contracted in iteration $t$,  that is, $e\in A[P_j]$ for a component $P_j$ of the partition $\mathcal{P}_t$. In particular, $c^\pi(e)=c^{p_t}(e)\ge L_t$. The set $P_j$  is strongly connected in the subgraph of arcs of reduced cost $\le L_t+2D_t\le \K L_t$. Thus, at iteration $t$, $P_j$ contains a cycle $C$ with $e\in C$ such that $\hat c_t(f)\le \K L_t$ for all $f\in C$.
This cycle may contain nodes that were contracted during previous iterations.  Every component previously contracted contains a strongly connected subgraph of arcs with costs $L_{t-1}+D_{t-1}\le \K L_{t-1}$, noting that the arc costs do not change anymore after contraction.
Thus, when uncontracting a node, we can extend $C$ to a cycle of arc costs $\le \K L_t$.
Hence, we can obtain a cycle $C'$ in the original graph $G$ with $e\in C'$ and $c^\pi(f)\le \K L_t\le \K c^\pi(e)$ for all $f\in C'$.

The algorithm only performs addition and comparison operations, and divisions by $2^\rho$. Divisions only happen when setting $D_{t+1}=L_{t+1}/2^\rho$. At such iterations, we have $D_{t+1}=2^k C_{\min}$, where $k=t/ 2^\rho$ is an integer.
As remarked in Section~\ref{sec:comp-mod}, we can implement every step in $O(\rho+1)$ time.
\end{proof}

\subsection{A quick algorithm for rough balancing}\label{sec:rough-balance}
In this section, we present the subroutine \textsc{Rough-balance}$(N,A,c)$, which finds a potential $\pi\in \R^N$ such that $c^\pi$ is  $\balval$-min-balanced. 
As mentioned previously, this will be an important preprocessing step for the strongly polynomial algorithm in Section~\ref{sec:balance-strong}. 
The running time can be stated as follows. Here, $\alpha(m,n)$ is the inverse Ackermann function. 
\begin{lemma}\label{lem:roughbal}
Let $G=(N,A,c)$ be a strongly connected directed graph with $c\in \R^N_{\po}$. Then, in time $O(m\alpha(m,n)\log n )$, we can find a potential $\pi\in \R^N$
such that $c^\pi$ is $\balval$-min-balanced, where $n=|N|$ and $m=|A|$. The algorithm can be implemented in the comparison-addition model, and every $c^\pi(e)$ value will be an integer multiple of $4n^2$.
\end{lemma}

 Given $G=(N,A,c)$ with $c\in \R^A_{\nn}$, and $r>0$, we let $G[\le r]$ 
  denote the subgraph of 
 $G$ formed by the arcs $e\in A$ with $c(e)\le r$.
For every $e\in A$, we define $\bal(e)\in \R_{\po}$ as the smallest value $r$ such that $G[\le r]$ contains a directed cycle $C$ with $e\in C$. We call $\bal(e)$ the \emph{balance value} of $e$. Clearly, $G$ is $\K$-min-balanced if and only if $\bal(e)\le \K c(e)$ for every $e\in A$.

The algorithm proceeds in two stages. Section~\ref{sec:find-balance} presents \textsc{Find-Balance}$(N,A,c)$, which determines the balance value $\bal(e)$ for every arc in $e\in A$. The main algorithm \textsc{Rough-Balance}$(N,A,c)$  follows in Section~\ref{sec:rough-alg}.

\subsubsection{Determining the balance values}\label{sec:find-balance}

Algorithm~\ref{alg:find-balance} presents the recursive subroutine \textsc{Find-Balance}$(N,A,c)$.
Let $c[1]<c[2]<\ldots<c[K]$ denote the set of different arc cost values. If $K=1$, i.e., all arc costs are the same, then we return $\bal(e)=c(e)=c[1]$ for every arc.
Otherwise, we let $r$ and $r'$ denote the two consecutive values in the middle.
We identify the strongly connected components $(N_1,A_1), (N_2,A_2),\ldots, (N_k,A_k)$ of $G[\le r]$, and recursively determine the $\bal(e)$ values for $e\in A_i$ by calling the algorithm for each nonsingleton component $(N_i,A_i)$. For all arcs  $e\in A[N_i]\setminus A_i$, we set  $\bal(e)= c(e)$. 

We then contract all components  $(N_i,A_i)$ to singletons to obtain $\hat G=(\hat N,\hat A,\hat c)$. We increase each arc cost $\hat c(e)$ in this graph to $\max\{\hat c(e),r'\}$, make another recursive call to the algorithm on $\hat G$, and use the obtained balanced values for the pre-images of the contracted arcs.
\begin{algorithm}[htb]
    \caption{\textsc{Find-balance}}\label{alg:find-balance}
    \begin{algorithmic}[1]
        \Require{A strongly connected directed graph $G=(N,A,c)$  with $c\in\Q_{\po}^A$.} 
                \Ensure{A function $\bal:A\to \Q$ giving the balance value $\bal(e)$ of each arc $e\in A$.}
                \State Let $c[1]<c[2]<\ldots<c[K]$ denote the set of arc cost values $c(e)$ ;
                \If{$K=1$} $\bal(e)\gets c(e)$ for all $e\in A$ ;
                \Else
                  \State $r\gets c\left[\left\lfloor\frac{K}2\right\rfloor\right]$ ; $r'\gets c\left[\left\lfloor\frac{K}2\right\rfloor+1\right]$ ;
                 \State $\{(N_1,A_1), (N_2,A_2),\ldots, (N_k,A_k)\}\gets$ \Call{Strongly-Connected}{$G[\le r]$} ;
     \For{$i=1,\ldots,k$} \If {$|N_i|>1$} 
       \For{$e\in  A[N_i]\setminus A_i$}  $\bal(e)\gets c(e)$ ; \EndFor
        \State $\bal_i\gets$\Call{Find-Balance}{$N_i, A_i, c|_{A_i}$}  ;
                \For {$e\in A_i$} $\bal(e)\gets \bal_i(e)$ ;\EndFor 
\EndIf \EndFor
 \State obtain $\hat G=(\hat N,\hat A,\hat c)$ by contracting  every set $N_j$, $j\in [k]$ ;
 \For {$e\in\hat A$} $\hat c(e)\gets \max\{\hat c(e),r'\}$ ;\EndFor
                \State $\hat \bal\gets$ \Call{Find-Balance}{$\hat G$} ;
                 \For {$e\in A\setminus \left(\bigcup_{j=1}^s A[N_j]\right)$} $\bal(e)\gets \hat \bal(\hat e)$, where $\hat e$ is the contracted image of $e$ ;\EndFor 
                \EndIf
    \State \Return{$\bal$.} 
    \end{algorithmic}
\end{algorithm}

\begin{lemma}\label{lem:balance-correct}
 Algorithm~\ref{alg:find-balance} correctly computes the balance values in $G$ in 
 time $O(m \log n)$.
\end{lemma} 
\begin{proof}
For any pair of nodes $v$ and $w$ in $N_i$,  $b(w, v) \le r$.  If 
$e = (v, w)\in A$ and $c(e) > r$, then $\beta(e) = c(e)$ is correctly determined.  If $c(e) \le r$, then $e \in A_i$, and $\beta(e)\le r$ can be found recursively by finding the balance values in $(N_i, A_i)$.

Suppose instead that $e = (v, w)\in E$, where $v \in N_i$ and $w \in N_j$ for $j\neq i$.  Then $\beta(e) \ge r'$, and we can replace $c(e)$ by 
$c'(e) = \max\{c(e), r'\}$ without changing $\beta(e)$.   In addition, contracting the strongly connected components $(N_i, A_i)$ does not affect 
$\beta(e)$.  Thus, the algorithm correctly computes the balance numbers.  
  
We now turn to the running time.   The initial time to sort the arcs is $O(m \log m) = O(m \log n)$.  Let $T(m', K')$ be the running time for the algorithm if the input has $m'$ sorted arcs with $K'$ different costs.  The algorithm finds the median value for the arc costs and partitions the graph into subgraphs with $m^*$ arcs and $m' - m^*$ arcs respectively, where 
$ m^* \in [1,m' - 1]$.   This takes $O(m')$ time, and leads to the following recursion:
\[
    T(m', K') \le O(m') +\max_{m^* \in [1,m']} \left\{T(m^*, \lfloor K/2\rfloor ) + T(m - m^*, \lfloor K/2\rfloor + 1)\right\}\, .  
\]
We conclude that $T(m', K') = O(m' \log(K' + 1))$ because the number of different arc values is halved in every recursive call.  Hence, every arc can participate in at most $\lfloor \log(K' + 1) \rfloor$ recursive calls.
\end{proof}

\subsubsection{Constructing the potential}\label{sec:rough-alg}
We now describe the algorithm \textsc{Rough-balance}$(N,A,c)$. 
We first compute the balance values $\bal(e)$ by running \textsc{Find-Balance}$(N,A,c)$. 
We define
\[
\eta(e):=\max\left\{c(e),\frac{\bal(e)}{2n}\right\}\, .
\]
For $r\ge 0$, we let $G[\eta\le r]$  denote the subgraph of 
 $G$ formed by the arcs $e\in A$ with $\eta(e)\le r$. We say that $e\in A$ is \emph{active} with respect to the value $r$ if $\eta(e)\le r$ but $e$ is not contained in any strongly connected component of $G[\eta\le r]$.
 
The \textsc{Rough-balance} subroutine is shown in Algorithm~\ref{alg:rough-balance}. A value $r\ge 0$ is maintained, and the graph $\hat G$ denotes the contraction of the strongly connected components of $G[\eta\le r]$; we use the $\eta$ values also in $\hat G$ that refer to the pre-image of the arc in $G$. 
At the beginning of the first iteration, $r$ is set as
the minimum $\eta(e)$ value in $G$;  in later iterations, we increase $r$ by a factor $2n$, or to the minimum of the $\eta(e)$ values in the current $\hat G$.
Each iteration computes a topological ordering of the  active arcs w.r.t.~$r$. Then, the potential $\pi_{v_i}$ of the $i$-th node $v_i$ in the order is decreased by $ri/(2n)$. 
 We terminate once $\hat G$ becomes a single node, i.e., $G[\eta\le r]$ is strongly connected.

 We handle contractions  as in Section~\ref{sec:simple-weakly}. That is, the final reduced cost of an arc $e$ is equal to its reduced cost immediately before its endpoints got contracted into the same node. At the end,  we uncontract and obtain the overall potential in the original graph in time $O(n)$.

 \begin{algorithm}[htb]
    \caption{\textsc{Rough-balance}}\label{alg:rough-balance}
    \begin{algorithmic}[1]
        \Require{A strongly connected directed graph $G=(N,A,c)$  with $c\in\R_{\po}^A$.} 
                \Ensure{A potential  $\pi:V\to \R$ such that $c^\pi$ is $\balval$-min-balanced.}
                \State obtain the balance values $\beta(e)$ by calling \Call{Find-Balance}{$G$} ;
                \For{$e\in A$} $\eta(e)\gets\max\left\{c(e),\frac{\bal(e)}{2n}\right\}\,$\, ; \EndFor
                \State $r\gets 0$ ; $\hat G\gets G$\, ;
                \For{$i\in N$} $\pi_i\gets 0$ ; \EndFor
                \While{$|\hat N|>1$}
                    \State $r\gets \max\{2nr,\min\{\eta(e):\, e\in \hat A\}\}$\; ;
                     \State contract all strongly connected components of $\hat G[\eta\le r]$ in $\hat G$ ;
                     \State compute a topological ordering  $\hat N=\{v_1,v_2,\ldots,v_k\}$ of $\hat G[\eta\le r]$ such that $i<j$ for all $(v_i,v_j)\in \hat A$ with $\eta(v_i,v_j)\le r$ ;
                     \For{$i=1,\ldots,k$} $\pi_{v_i}\gets \pi_{v_i}-\frac{ri}{2n}$ ;\EndFor
                \EndWhile
                \State uncontract $\hat G$ and map $\pi$ back to the original graph $G$ ;
                    \State \Return{$\pi$.} 

    \end{algorithmic}
\end{algorithm}

We now turn to the proof of Lemma~\ref{lem:roughbal}. As the first step, we bound the reduced costs obtained in the algorithm. The reduced costs are defined in the contracted graph, but can be naturally mapped back to the input graph $G$.

\begin{lemma}\label{lem:pot-change}
Consider the potentials at the end of any iteration of Algorithm~\ref{alg:rough-balance}, and let $c^\pi(e)$ denote the reduced cost of any arc $e\in A$.
  Then, $|c^\pi(e)-c(e)|\le 2r/3$. If $e$ is an active arc in the current iteration, then $c^\pi(e)\ge c(e)+r/(6n)$.
\end{lemma}
\begin{proof}
The initial potential values are $\pi_i=0$ and are monotone decreasing throughout. The current iteration decreases every potential by at most $r/2$. Since the value of $r$ increases by at least a factor $2n\ge 4$ in every iteration, the cumulative change in all iterations thus far is at most $2r/3$. This implies the first statement. 

Assume now that $e=(v_i,v_j)$ is active. Then, $c^\pi(e)$ increases by at least $r/(2n)$ in the current iteration, since $\pi_{v_i}$ is decreased by a smaller amount than $\pi_{v_j}$. The second part follows, since the total change up to the previous iteration with value $r''\le r/(2n)$ was $2r''/3\le r/(3n)$.
\end{proof}

 \begin{lemma}\label{lem:strongconn-bal}
 An arc $e\in A$ is contained in a strongly connected component of $G[\eta\le r]$ if and only if $\beta(e)\le r$. Every arc can be active in at most one iteration.
 \end{lemma}
 \begin{proof}
 Suppose first that $\beta(e)\le r$.
 By definition of $\beta(e)$, there exists a cycle $C$ with $e\in C$ such that
$c(f)\le \beta(e)$  for all $f\in C$. 
Consequently, $\beta(f)\le \beta(e)$ and $\eta(f)\le \beta(e)$ for all $f\in C$, showing that $e$ is inside a strongly connected component of $G[\eta\le r]$ whenever $\beta(e)\le r$. Conversely, assume that  there exists a cycle $C'$ containing $e$ such that $\eta(f)\le r$ for all 
 $f\in C'$. Since $c(f)\le \eta(f)$, it follows that $\beta(e)\le r$. 

 Therefore, an arc $e$ is active if and only if $\eta(e)\le r<\beta(e)$. Since $\eta(e)\ge \beta(e)/(2n)$, and $r$ increases by at least a factor $2n$ between two iterations, it follows that each arc can be active at most once.
 \end{proof}

\begin{proof}[Proof of Lemma~\ref{lem:roughbal}]
We first show that the algorithm \textsc{Rough-balance} finds a $\balval$-min-balanced cost function.
Consider any arc $e\in A$, and let us pick a cycle $C$ containing $e$ such that $c(f),\bal(f)\le \bal(e)$ for every $f\in C$. Take the largest value of $r$ during the algorithm such that $r<\bal(e)$; let $r'\ge \bal(e)$ denote the value in the next iteration. By Lemma~\ref{lem:strongconn-bal}, $e\in\hat A$ in the current iteration, and $e$ will be contracted in the next iteration, along with the entire cycle $C$. Hence, $|c^\pi(f)-c(f)|\le 2r/3$ for all $f\in C$ for the final reduced cost $c^\pi$  according to Lemma~\ref{lem:pot-change}. 

\begin{claim}\label{cl:rp}
We have  $r'=2nr$ or $r'=c(e)=\beta(e)$.
\end{claim}
\begin{proof}
If $r'>2nr$, then $r'=\min\{\eta(f):\, f\in \hat A\}$. Hence, $r'\le \eta(e)\le \beta(e)\le r'$. 
Equality must hold throughout, which in particular implies $\eta(e)=c(e)=\beta(e)$.
\end{proof}
We consider two cases.

\paragraph{Case I:  $r<c(e)$.}  By the above claim, $\bal(e)\le r'\le 2n c(e)$.
On the one hand, we have
\[
c^\pi(e)\ge c(e)-\frac{2}{3}r\ge \frac13 c(e)\, .
\]
On the other hand, for every $f\in C$, we have
\[
c^\pi(f)\le c(f)+\frac23r\le \bal(e)+\frac23 c(e)\le \left(2n+\frac23\right) c(e)\le (6n+2) c^\pi(e)\, .
\]
Hence, $\beta^\pi(e)\le (6n+2) c^\pi(e)<14n^2 c^\pi(e)$.
\paragraph{Case II:  $r\ge c(e)$.}  Since $r<r'$, 
Claim~\ref{cl:rp} yields
$r'=2nr\ge \beta(e)$, and thus $r\ge\beta(e)/(2n)$.
 Consequently, $r\ge \eta(e)=\max\{c(e),\beta(e)/(2n)\}$.

Since $\eta(e)\le r<\beta(e)$, by Lemma~\ref{lem:pot-change}, $e$ is an active arc, and the second part of the lemma guarantees that 
\[
c^\pi(e)\ge c(e)+\frac r{6n}\, .
\]
For every $f\in C$, we have
\[
c^\pi(f)\le c(f)+\frac23r\le \bal(e)+\frac23 r\le \left(2n+\frac23\right)r\le \left(2n+\frac23\right)6n c^\pi(e)\le 14n^2 c^\pi(e)\, 
\]
showing that  $\beta^\pi(e)\le 14n^2 c^\pi(e)$. This completes the proof.

\paragraph{Running time bound} 
The initial call to \textsc{Find-Balance}$(N,A,c)$ takes 
$O(m\log n)$ time according to Lemma~\ref{lem:balance-correct}.
The significant terms in the running time are computing strongly connected components of $G[\eta\le r]$ along with the topological ordering of active arcs, and updating the potentials. According to Lemma~\ref{lem:strongconn-bal}, each arc is active at most once. Hence, it is either contracted in the first iteration it appears in 
$G[\eta\le r]$, or the subsequent one. Therefore, the total number of these operations is $O(m)$.
Maintaining the contracted graph using the Union-Find data structure is $O(m\alpha(n,m))$, see also Section~\ref{sec:update-costs}.
The number of operations in  the final uncontraction is $O(n)$, similarly to the argument in Section~\ref{sec:simple-weakly}.

To implement in the comparison-addition model, note that every number during the computations will be integer multiples of $(2n)^2$. Additions, subtractions, and comparisons of numbers in this form can be implemented in $O(\log n)$ time, as in Section~\ref{sec:comp-mod}. We also need multiplications by $i\le n$ and by $2n$ as well as divisions by $2n$; these operations also take time $O(\log n)$. Hence, the total running time can be bounded by $O(m\alpha(m,n)\log n)$.
\end{proof}

\subsection{The strongly polynomial algorithm}\label{sec:balance-strong}
We are ready to present the strongly polynomial algorithm as stated in Theorem~\ref{thm:balance-main}.
 Given a graph $G=(N,A,c)$ with nonnegative arc costs, we  preprocess it by contracting 0-cycles and changing to a strictly positive reduced cost. We then apply the subroutine \textsc{Rough-balance} to find a  $\balval$-min-balanced reduced cost function $c^\pi$.
  We can thus assume that the input of 
Algorithm~\ref{alg:balance} is a strictly positive and $\balval$-min-balanced cost function $c$.


Algorithm~\ref{alg:balance} is similar to the weakly polynomial Algorithm~\ref{alg:weakly}. The two crucial differences are that {\em (a)} the subroutine \textsc{Small-Cycles} is called only for a subset of `active' arcs; and {\em (b)} we
may `jump' over irrelevant values of $L$.


\begin{algorithm}[htb]
    \caption{\textsc{Min-Balance}}\label{alg:balance}
     \begin{algorithmic}[1]
        \Require{A strongly connected directed graph $G=(N,A,c)$  with a $\balval$-balanced cost vector $c\in\R_{\po}^A$, parameters $\rh\in\Z_{+}$ and $\K=1+1/2^{\rh-1}$.}
                \Ensure{A potential vector $\pi\in \R^N$ such that $c^\pi$ is  $\K$-min-balanced.}
                \State sort all arcs in the increasing order of costs as $c(e_1)\le c(e_2)\le\ldots\le c(e_m)$ ;
                \State $(\hat N_1,\hat A_1,\hat c_1)\gets (N,A,c)$ ; $t\gets 1$ ;
                \State $L_1\gets c(e_1)$, $D_1\gets L_1/2^\rh$ ; \label{l:starting-val}
                \State $F_1\gets \left\{e\in A: c(e)\le (n+1)\left(1+\frac{1}{2^\rho} \right)L_1\right\}$ ;
                \While{$|\hat N_t|>1$}
                    \State $(\mathcal{P}_t,p_t)\gets $ \Call{Small-Cycles}{$L_t,D_t,\hat  N_t(F_t),F_t,\hat c_t$} ;
                    \State  $(\hat N_{t+1},\hat F, \hat c_{t+1})\gets (\hat N_t,F_t,\hat c_t^{p_t})/\mathcal{P}_t $ ;\label{l:contract}
                    \State $L_{t+1}\gets L_t+D_t$ ;
                    \If{$t$ is an integer multiple of $2^\rh$} 
                     \If{$4nL_{t+1}<\min_{e\in\hat A_{t+1}}  c(e)$} $L_{t+1}\gets \min_{e\in\hat A_{t+1}}  c(e)/2$ ;\label{l:update-L}\EndIf 
                    \State $D_{t+1}\gets L_{t+1}/2^\rh$ ; \label{line:D-set}
                    \Else\ 
                     $D_{t+1}\gets D_t$ ;\EndIf  
                    \State $F_{t+1}\gets \hat F\cup \left\{e\in \hat A_{t+1}: (n+1)\left(1+\frac{1}{2^\rho} \right)L_{t}< c(e)\le (n+1)\left(1+\frac{1}{2^\rho} \right)L_{t+1}\right\}$ ;\label{l:update-F}
                    \For {$e\in F_{t+1}\setminus F_t$}
                    $\hat c_{t+1}(e)\gets$\Call{Get-Cost}{$e$} ;\EndFor \label{l:get-cost}
                \EndWhile
                \State uncontract $(\hat N_t,\hat A_t,\hat c_t)$, and compute the overall potential $\pi\in \R^N$ .\label{line:uncontract}
    \State \Return{$\pi$.} 
    \end{algorithmic}
\end{algorithm}

At the beginning of the algorithm, we sort the arcs in the increasing order of costs $c(e)$. 
At iteration $t$, we maintain two key parameters, the \emph{`lower bound'} $L_t$ and the \emph{`step-size'} $D_t$, a contracted graph $(\hat N_t,\hat A_t,\hat c_t)$, and a set of \emph{active arcs} $F_t\subseteq \hat A_t$. This is the subset of arcs with $\hat c_t(e)\le (n+1)\left(1+\frac{1}{2^\rho} \right) L_t$. 

As in Algorithm~\ref{alg:weakly}, the parameters are initialized as $L_1=c(e_1)$, $D_1=L_1/2^\rho$. The iterations start with a call to \textsc{Small-Cycles} for the current value of $L_t$ and $D_t$, but restricted to the graph $(\hat N_t(F_t),F_t)$ induced by the active arcs; this returns a partition $\mathcal{P}_t$ and potentials $p_t$.
With a slight abuse of notation, the node potentials $p_t$ are extended to the entire node set $\hat N_t$, by setting $p_t(v)=0$ for $v\in \hat N_t\setminus \hat N_t(F_t)$.
 We contract each non-singleton subset  in the partition $\mathcal{P}_t$; the new costs $\hat c_{t+1}$ represent the contractions of $\hat c_t^{p_t}$. However, we only maintain the $\hat c_{t+1}(e)$ values explicitly for the active arcs $\hat F$, the contracted image of $F_{t}$.

We now turn to the updates of $L_t$ and $D_t$.
In most iterations\footnote{More precisely, in $1-2^{-\rho}$ fraction of all iterations;  there are no such iterations for $\rho=0$.}, we set $L_{t+1}=L_t+D_t$, and keep $D_{t+1}=D_t$. Exceptions are the special iterations when $t$ is an integer multiple of $2^\rho$, in which case we set $D_{t+1}=L_{t+1}/2^\rho$.
In these special iterations, the update defining $L_{t+1}$ is also different. We start by letting $L_{t+1}=L_t+D_t$, and then compare this value to $\min_{e\in \hat A_{t+1}} c(e)/(4n)$.  If $L_{t+1}$ is smaller, then we increase $L_{t+1}$ to $\min_{e\in \hat A_{t+1}} c(e)/2$.  
Note that $L_t$ increases by at least a factor 2 between any two special iterations.

After updating $L_{t+1}$ and $D_{t+1}$, we 
update the set of active arcs by adding all arcs $e\in \hat A_{t+1}$ with cost $c(e)\in \left((n+1)\left(1+\frac{1}{2^\rho} \right) L_t,(n+1)\left(1+\frac{1}{2^\rho} \right)L_{t+1}\right]$. We emphasize that $c(e)$ here refers to the input costs and not the reduced cost. 
The subroutine \textsc{Get-Cost}$(e)$ obtains the reduced cost $\hat c_t(e)$ of the newly added arcs. This will be explained in Section~\ref{sec:update-costs}, using the \emph{Union-Find-Increase} data structure.
We terminate once the graph is contracted to a singleton; at this point, we uncontract and obtain the output potential $\pi$ in the original graph as in Algorithm~\ref{alg:weakly}.

\medskip

Let us now turn to the analysis. We let $T$ denote the total number of iterations.

\begin{lemma}\label{lem:bound-change}
Let $\tau\in [T]$ be an iteration of Algorithm~\ref{alg:balance} such that in all previous iterations $t\in [\tau]$,
$\hat c_{t}(e)\ge L_{t}$ was valid for all $e\in F_t$. Then, $|\hat c_{\tau+1}(e)-c(e)|\le n\left(1+\frac{1}{2^\rho} \right) L_{\tau}$ for every $e\in \hat A_t$.
\end{lemma}
\begin{proof}
The condition guarantees that the input to \textsc{Small-Cycles} at all iterations $t\le \tau$ satisfies the requirement on the arc costs.
The potential $p_t$ found by \textsc{Small-Cycles}  has values $-|\hat N_t| D_t\le  p_t(v)\le 0$. Therefore, for each $e\in \hat A_t$, $|\hat c_{\tau+1}(e)-c(e)|\le n \sum_{t=1}^{\tau} D_t$.

We show that $\sum_{t=1}^{\tau} D_t\le \left(1+\frac{1}{2^\rho} \right) L_{\tau}$.
Indeed, $L_{t+1}\ge L_t+D_t$ in every iteration, implying 
$\sum_{t=1}^{\tau-1} D_t\le L_\tau$; and $D_{\tau}\le L_\tau/2^\rh$.
\end{proof}

\begin{lemma}\label{lem:strongly-correct}
In every iteration $t \in [T]$ of Algorithm~\ref{alg:balance}, $\hat c_t(e)\ge L_t$ for all $e\in \hat A_t$. The final reduced cost function $c^\pi$ is $\K$-min-balanced. Further, every arc $e\in A$ with $c(e) <L_t/(14n^3)$ was contracted before iteration $t$.
\end{lemma}
\begin{proof}
Let us start with the first claim. The proof is by induction. For $t=1$, $\hat c_1(e)\ge L_1$ is true for every $e\in A=\hat A_1$ by the definition of $L_1=\ c(e_1)$. Assume the claim was true for all $1\le t'\le t$; we show it for $t+1$. 

Assume first we set the value  $L_{t+1}=\min_{f\in\hat A_{t+1}} c(f)/2$ in an iteration where $t$ is divisible by $2^\rho$.
This happens if $2n(L_t+D_t)<\min_{f\in\hat A_{t+1}} c(f)/2$.
 Lemma~\ref{lem:bound-change} then implies that $\hat c(e)>\min_{f\in\hat A_{t+1}} c(f)- 2nL_t>L_{t+1}$ for every $e\in \hat A_{t+1}$. 

 Let us next assume the update was $L_{t+1}=L_t+D_t$. If $e\in \hat F$, i.e., the contracted image of $F_{t}$, then $\hat c_{t+1}(e)\ge L_{t}+D_{t}=L_{t+1}$ is guaranteed by \textsc{Small-Cycles}. Let $e\in \hat A_{t+1}\setminus F_t$, i.e., $c(e)>(n+1)\left(1+\frac{1}{2^\rho} \right) L_t$. Then, 
Lemma~\ref{lem:bound-change} shows  $\hat c_{t+1}(e)>\left(1+\frac{1}{2^\rho} \right)L_t\ge L_{t+1}$.

The $\K$-min-balancedness property of the final reduced cost $c^\pi$ follows as in Lemma~\ref{lem:weakly-analysis} for the weakly polynomial Algorithm~\ref{alg:weakly}. 

Consider now an arc $e\in A$ with $c(e)<L_t/(14n^3)$. By the $\balval$-min-balancedness of the input cost function $c$, there exists a cycle $C\subseteq A$ such that $c(f)\le \balval c(e)$ for all $f\in C$. The final reduced cost $c^\pi$ is nonnegative, and therefore
\[
c^\pi(e)\le c^\pi(C)=c(C)\le 14n^3 c(e)< L_t\, .
\]
Recall that the final reduced cost $c^\pi(e)$ equals 
$\hat c_{t'} (e)$ for the iteration $t'$ when $f$ was contracted. Since $\hat c_t(f)\ge L_t$ for all $f\in \hat A_t$, it follows that $t'<t$, as required.
\end{proof}

In Section~\ref{sec:update-costs} we will show that the overall running time of the operations \textsc{Get-Cost}$(e)$ can be bounded as $O(m\alpha(m,n))$. We need one more claim that shows the geometric increase of $L_t$.

\begin{lemma}\label{lem:L-t-change}
For every iteration $t'\ge 1$, we have $L_{t'+2^\rh}\ge 2L_{t'}$. 
\end{lemma}
\begin{proof}
Let $t=t'+2^\rh$.
Assume first $2^\rh|t'-1$. Then, $D_{t'}=L_{t'}/2^\rh$, and we have $D_{t''}=D_{t'}$ for all $t''\in [t',t-1]$. Consequently, $L_{t}\ge L_{t'}+2^\rh D_{t'}=2L_{t'}$. The inequality may be strict if in iteration $t-1$ we set $L_{t}>L_{t-1}+D_{t-1}$. 

Assume now $t'=t_0+k$ such that $2^\rh|t_0-1$ and $k\in [1,2^\rh-1]$. Then, $L_{t'}=L_{t_0}(1+k/2^\rh)$, $L_{t_0+2^\rh}\ge 2L_{t_0}$, and $L_t=L_{t_0+2^\rh}(1+k/2^\rh)\ge 2L_{t_0}(1+k/2^\rh)$, thus, we again have $L_{t}\ge 2L_{t'}$.
\end{proof}
We are ready to prove Theorem~\ref{thm:balance-main}.
\begin{proof}[Proof of Theorem~\ref{thm:balance-main}] {\em Part \eqref{part:min-bal-alg}: the approximate min-balancing algorithm.} Let us start with bounding the total number of arithmetic operations.
After the $O(m\log n)$ preprocessing algorithm, we run the algorithm \textsc{Rough-Balance} to find a $\balval$-balanced cost function in time $O(m\log n)$ (Lemma~\ref{lem:roughbal}). We now turn the analysis of Algorithm~\ref{alg:balance}. Let $m_t=|F_t|$ denote the number of active arcs in iteration $t$. The number of arithmetic operations in \textsc{Small-Cycles} in iteration $t$ is bounded as $\max\{O(1),O((\rho+1)\cdot m_t\sqrt{n})\}$. The term $O(1)$ is needed since there may be some `idle' iterations without any active arcs, that is, $m_t=0$. In such a case the update rule  in line~\ref{l:update-L} guarantees that new active arcs appear within 
the next $O(2^\rho)$ iterations. Thus, the number `idle' iterations without active arcs can be bounded as $O(m2^\rho)$, since every arc can give the minimum value in line~\ref{l:update-L} at most once. The total running time  of the `idle' iterations is dominated by the other terms. 

Let us now focus on the iterations containing active arcs. We show that
\begin{equation}\label{eq:m-bound}
\sum_{t=1}^T m_t = O\left(2^\rh {m \log n}\right)\, .
\end{equation}
Consider any arc $e\in A$. Let $t_1$ be the first and $t_2$ be the last iteration such that $e\in F_t$. By definition, $t_1$ is the smallest value such that  $c(e)\le (n+1)\left(1+\frac{1}{2^\rho} \right) L_{t_1}$, and by the last part of Lemma~\ref{lem:strongly-correct} $L_{t_2}/(14 n^3) \le c(e)$. Thus, $L_{t_2}\le 28n^4 L_{t_1}$. Lemma~\ref{lem:L-t-change} shows that $L_t$ increases by a factor 2 in every $2^\rh$ iterations. Hence, 
$t_2-t_1\le 2^\rh\log (28n^4)$, implying \eqref{eq:m-bound}.

Hence, the total number of operations in the calls to \textsc{Small-Cycles} is bounded as $O\left(2^\rho m \sqrt{n}\log n\right)$. The time of contractions and cost updates can be bounded 
as $O(m\alpha(m,n))$ as shown in Section~\ref{sec:update-costs}, and the final uncontraction takes $O(n)$. 

\medskip

\noindent\emph{Implementation in the comparison-addition model:} 
As noted previously, \textsc{Rough-Balance} and \textsc{Small-Cycles} are both implementable in this model.
 Algorithm~\ref{alg:balance} uses additions, comparisons, multiplications by $4n$, divisions by $2$ and by $2^\rho$. Further, all numbers in the computations will be integer multiples of $2^b$ for $b\le {2\rho+1}$. As noted in Section~\ref{sec:comp-mod}, all operations can be implemented in time $O(\rho+1)$. The running time bound follows.

\paragraph{Part \eqref{part:comp-hi}: obtaining the component hierarchy.} 
Assume now $\rho=0$ and $\K=3$; let us use the algorithm as described in Algorithm~\ref{alg:balance} with two simple modifications: we set the initial value as $L_1=\lfloor c(e_1)\rfloor_2$ in line~\ref{l:starting-val}, and if 
$4nL_{t+1}<\min_{e\in\hat A_{t+1}} c(e)$,  then we update $L_{t+1}$ to $\left\lfloor\min_{e\in \hat A_{t+1}}c(e)/2\right\rfloor_2$ in line~\ref{l:update-L}. Thus, these values are rounded down to the nearest power of two. Such an operation is not allowed in the comparison-addition model, but can be done by a most significant bit operation in the word RAM model.

Recalling also that $n$ is a power of 2, and that we set $D_{t+1}=L_{t+1}/2^\rho=L_{t+1}$ in every step, it follows that every $L_t$ value is a power of 2.

The sets contracted during the algorithm can be naturally represented by a rooted tree $(V\cup N,E)$, where the nodes $N$ correspond to the leaves and the root $r\in V$ to the final contraction of the entire node set. If the set represented by some $v\in V$ was contracted at iteration $t$, we set $a(v)=L_t$.

We claim that $(V\cup N,E,a)$ forms a component hierarchy of $G^\pi=(N,A,c^\pi)$.
 All $a(v)=L_t$ values are integer powers of $2$ (this is the reason for the additional rounding steps).
It is immediate that the leaves in the subtree of each $v\in V$ form a strongly connected component in $G^\pi$. Let $v$ represent a set contracted in iteration $t$, that is, $v=P_i$ for a set $P_i$ in the partition $\mathcal{P}_t$. 
If $\lca(i,j)=v$ for $i,j\in N$, that means that the nodes $i$ and $j$ got contracted together in iteration $t$. We show that $a(v)\le \beta(i,j)\le 3 a(v)$, and that the nodes in $\desc(v)$ contain a  path between $i$ and $j$ of arcs with cost at most $3a(v)$; consequently, $\beta(i,j)\le 3a(v)$. If $t=1$, then $L_1=\lfloor C_{\min}\rfloor_{2}$, and $P_i$ is strongly connected in the subgraph of arcs of cost at most $3 L_1$. If $t>1$, then $(\hat N_t,\hat A_t)$ contains a path between the contracted images of $i$ and $j$ with all arc costs between $a(v)=L_t$ and $3 L_t$, and every $i$--$j$ path must contain an arc of cost $\ge L_t$. We can map this back to the original graph by uncontracting the sets from previous iterations; all arc obtained in the uncontraction will have costs $\le 3 L_{t-1}< 3 L_t$.

Note that for $\rh=0$ and an integer input, the Algorithm~\ref{alg:balance} finds an integer $\pi$. This is because all $D_t$ values are integral, and \textsc{Small-Cycles} changes the potential by integer multiples of $D_t$. However, the input to Algorithm~\ref{alg:balance} is not the original cost but the  cost obtained after the preprocessing and \textsc{Rough-Balance}. Preprocessing returns $nc^{\bar\pi}$ for a $1/n$-integral potential $\bar\pi$. For an integer input $c$, \textsc{Rough-Balance} returns a $1/4n^2$-integral potential. From these three steps, we can obtain a relabelling $c^\pi$ of the original potential that is $1/(4n^3)$-integral if the original input cost was nonnegative integer.
\end{proof}

\subsection{Union-Find-Increase: Maintaining the reduced costs}\label{sec:update-costs}
In \textsc{Get-Cost}$(e)$, we need to compute the current reduced cost of an arc $e$.
Let $e=(i,j)$ in the original graph. In the current contracted graph $\hat N_t$, $e$ is mapped to an arc $(i',j')$; that is, $i$ is in a contracted set represented by node $i'$, and $j$ is in a contracted set represented by $j'$ ($i=i'$ and $j=j'$ is possible).  In the case that $e$ is newly active (that is, it was not active at the previous iteration), we need to  recover the reduced cost $\hat c_t(e)$.  We do so by performing the uncontractions, as in the final step. Let $\pi$ be the potential obtained by uncontracting all sets. To compute $\pi(i)$, we need to 
add up all the $p_{t'}(i[t'])$ values for every iteration $t'\le t$, where $i[t']$ is the contracted node in $\hat N_{t'}$ representing $i$, and similarly for computing $\pi(j)$. Since there could have  already been $\Omega(t)=\Omega(n)$ contractions of sets containing $i$ and $j$, a na\"\i ve implementation would take $O(n)$ to compute a single reduced cost, or $O(nm)$ to obtain all current reduced costs.

We show that the time to calculate the reduced costs of newly active arcs in  \textsc{Small-Cycles} can be bounded as $O(m\alpha(m,n))$ by using an appropriate variant of the classical \emph{Union-Find} data structure that we call \emph{Union-Find-Increase}.

\medskip

 We refer the reader to \cite{tarjan1983} and \cite[Chapter 21]{CLRS} for the description and analysis of  \emph{Union-Find}; we highlight the simple modifications only.
The data structure maintains a forest $F$ on the node set $N=\{1,2,\ldots,n\}$, with each tree in $F$ corresponding to a set in the partition.
For each $i\in N$, let 
 Anc$(i)$ be the ancestors of $i$ in $F$ (including $i$).

In addition, each $i\in N$ is associated with a key value $\sigma(i)$ that is initially $0$, and which changes dynamically. 
We add two new operations to the data structure  \emph{Union-Find}: 
the operation $\textrm{Increase}(i, \delta)$
increases $\sigma(j)$ by $\delta$ for all $j$ in the same tree as $i$; and the operation $\textrm{Value}(i)$ returns $\sigma(i)$.
However, the $\sigma(i)$ values are not maintained explicitly. Instead, the algorithm maintains  auxiliary
values $\tau(j)$ such that the following property is satisfied for all $i\in N:$
\begin{equation}\label{eq:sigma-tau}
\sigma(i)=\sum_{j\in\textrm{Anc}(i)}\tau(j)\, .
\end{equation}

We need to modify the original operations as follows:
\begin{itemize}
\item    Suppose a \textrm{Union} operation is performed on root nodes $j$ and $k$, and $j$ is made the root of the combined component.  Then $\tau(k) \gets \tau(k) - \tau(j)$.
\item  Suppose that a path compression takes place along path $j_1, \ldots, j_k$, where $j_k$ is the root of the nodes in $j_1$ to $j_k$.   The \textsc{Union-Find} algorithm sets the parent of $j_i$ to $j_k$ for $i \in[1,k-1]$.    Let $\gamma(j_i) = \tau(j_{i+1}) + \ldots + \tau(j_{k-1})$;  the time to compute the values are proportional to the length of the path.   In addition to compressing the path, we set $\tau(j_i) \gets \tau(j_i)+\gamma(j_i)$  for each $i\in [1, k-1]$.  
\end{itemize}

Given these modifications, $\textrm{Increase}(i, \delta)$ can be implemented by first calling \textrm{Find}$(i)$ to determine the root $j$ of the tree containing $i$, and increasing $\tau(j)$ by $\delta$.
To implement $\textrm{Value}(i)$, we first run Find$(i)$, which uses path compression so that $i$ becomes the child of the root node  node $j$ of the tree. Thus, we can return $\sigma(i)=\tau(i)+\tau(j)$.

Clearly, the amortized complexity bound $O(\ell \alpha(\ell, n))$ for a sequence of $\ell$ steps for \emph{Union-Find} is applicable for the modified data structure.

\medskip

When applying \emph{Union-Find-Increase} to implement the operations \textsc{Get-Cost}$(e)$,
the key values $\sigma(i)$ correspond to the uncontracted potentials $\pi(i)$, and the sets to the pre-images of the nodes $v\in \hat N_t$ in the original node set $N$. We can further contract sets with the $\textrm{Union}$ step. When $p_t(v)$ is changed by $\delta$ for a contracted node $v\in \hat N_t$, we need to update the potential of every original node represented by $v$; this is achieved by $\textrm{Increase}(i, \delta)$. Finally,  \textsc{Get-Cost}$(e)$ for an arc $e=(i,j)$ can be implemented by calls to $\textrm{Value}(i)$ and $\textrm{Value}(j)$, and setting $\hat c_t(e)=c(e)+\pi(i)-\pi(j)$.   

\subsection{The adaptation of Goldberg's algorithm}\label{sec:goldberg}
In this section, we prove 
 Lemma~\ref{lem:small-cycles}, showing how the subroutine \textsc{Small-Cycles} can be implemented using a modification of Goldberg's algorithm \cite{Goldberg1995}.

Let $G=(N,A,c)$ be a directed graph with an integer cost function $c\in \Z^A$. Let $n=|N|$, $m=|A|$, and $C=\|c\|_\infty$.
Goldberg  developed an $O(m\sqrt{n}\log C)$ algorithm
 that finds a shortest path in a network or else finds a negative cost cycle.  The algorithm runs in $\log C$ scaling phases. The key subroutine is \textsc{Refine}; this is called at each scaling phase and takes $O(m\sqrt{n})$ time.  

\begin{algorithm}
    \caption{\textsc{Refine}}
    \begin{algorithmic}[1]
        \Require{A directed graph 
$G=(N,A,c)$ with a cost function $c \in \Z^A$ such that $c(e)\ge -1$ for all $e\in A$.}
                \Ensure{ A negative cost cycle $C$, or a potential vector $\pi\in\Z^N$ such that 
 \begin{enumerate}[(i)]
\item$c^\pi(e)\ge 0$
 for all $e\in A$, and
 \item $-n+1\le  \pi(v)\le 0$ for every $v\in N$.
\end{enumerate}}
    \end{algorithmic}
\end{algorithm}
 
We  describe the modification \textsc{Balanced-Refine} that allows for negative cost cycles in a specific way. 

\begin{algorithm}
    \caption{\textsc{Balanced-Refine}}
    \begin{algorithmic}[1]
        \Require{A directed graph 
$G=(N,A,c)$ with a cost function $c \in \Z^A$ such that $c(e)\ge -1$ for all $e\in A$.}
                \Ensure{ A potential vector $\pi\in\Z^N$ and a subset of arcs $A'\subseteq A$ such that 
\begin{enumerate}[(i)]
\item \label{p:inside-P} $A'$ is the union of directed cycles, and 
$-1\le c^\pi(e)\le 0$ for all $e\in A'$;
\item \label{p:between} $c^\pi(e)\ge 0$
 for all $e\in A\setminus A'$, and
 \item  \label{p:bound} $-n+1\le  \pi(v)\le 0$ for every $v\in N$.
\end{enumerate} }
    \end{algorithmic}
\end{algorithm}

The running time of \textsc{Balanced-Refine} is also $O(m\sqrt{n})$.
The subroutine \textsc{Small-Cycles}  (see Lemma~\ref{lem:small-cycles})  calls this for the cost function 
$\bar c(e)=\left\lfloor \frac{c(e)-L}{D}\right\rfloor -1$. We  obtain an arc set $A'$  and a  potential $\bar \pi$. We return the partition $\mathcal{P}$ formed by the (strongly) connected components of $A'$, and the potential $\pi=D\bar \pi$.
Note that  $\bar c^{\bar \pi}(e)\le 0$ implies $L\le c^\pi(e)\le L+2D$, and $\bar c^{\bar \pi}(e)\ge 0$ implies $c^\pi(e)\ge L+D$. The required properties then follow.

To obtain an algorithm in the comparison-addition model, we do not need to compute the $\bar c(e)$ values explicitly: the only relevant information will be whether an arc cost is $-1$, $0$, or positive. This simply corresponds to the cases $L\le c(e)< L+D$, $L+D\le c(e)\le L+2D$, and $L+2D< c(e)$. We can directly update the original potentials $\pi$, subtracting $Dk$ whenever $\bar\pi$ is decreased by $k$. This leads to an overhead $O(\log |N|)$ in the overall running time. 

\medskip
For completeness, we now describe the subroutines \textsc{Refine} and \textsc{Balanced-Refine} in parallel; omitted parts of the analysis follow as in \cite{Goldberg1995}.
%
%
Both algorithms iteratively construct an integer potential $\pi\in \Z^N$.   Throughout, $c^\pi(e)\ge -1$ for all $e\in A$.  At termination, $c^\pi(e)\ge 0$ for all arcs in the contracted graph.
The main difference is that \textsc{Refine} terminates once a negative cycle is found. In contrast, \textsc{Balanced-Refine} adds all negative cycles to the arc set $A'$ and contracts them.

 An arc with $c^\pi(e)\le 0$ is called \emph{admissible}; we let $G_\pi=(N,A_\pi)$ be the subgraph formed by admissible arcs. Arcs with $c^\pi(e)=-1$ are called \emph{improvable arcs}, and nodes with incoming improvable arcs are called \emph{improvable nodes}; we denote this set as $I\subseteq N$.

The \textsc{Decycle} subroutine  eliminates all directed cycles from $G_\pi$ by contractions, using \textsc{Strongly-Connected}$(G_\pi)$. \textsc{Refine} terminates if a negative cost cycle is found; in contrast, \textsc{Balanced-Refine} adds all such cycles to $A'$ and proceeds with the algorithm.
Contractions are carried out as described in Section~\ref{sec:balance}.

A set of nodes $S\subseteq N$ is \emph{closed} if no admissible arc leaves $S$.
For a closed set, the subroutine \textsc{Cut-Relabel}$(S)$ decreases $\pi(u)$ by 1 for every $u\in S$. The closedness of $S$ guarantees that no improvable arcs are created.

Assume $G_\pi$ is acyclic. Let us pick any improvable node
 $i$, and let $S$ be the set of nodes reachable from $i$ in $G_\pi$; this is a closed set.
 After \textsc{Cut-Relabel}$(S)$, $i$ is no longer improvable, and no new improvable nodes appear.
In this manner, we can decrease the number of improvable nodes in $O(m)$ time.
By alternating between the subroutines \textsc{Decycle} and \textsc{Cut-Relabel}, one can eliminate all improvable nodes in $O(nm)$ time, resulting in a graph with nonnegative reduced costs.   

\medskip

Goldberg improves this to $O( m\sqrt{n})$ time by eliminating at least $\sqrt{k}$ improvable nodes in 
$O(m)$ time, where $k=|I|$ is the number of improvable nodes in $G_\pi$.  
The first step in speeding up the running time is to eliminate more than one improvable node when running \textsc{Cut-Relabel}.  

A set $X\subseteq I$ of improvable nodes is called an \emph{anti-chain} in $G_\pi$ if for all nodes $i$ and $j$ in $X$, there is no directed path from node $i$ to node $j$ in $G_\pi$.    
Let $S$ be the set of nodes reachable in $G_\pi$  from a node of $X$.  After 
running \textsc{Cut-Relabel}$(S)$, none of the nodes in $X$ are improvable. 

In order to find a large anti-chain of improvable nodes, Goldberg's algorithm  appends a source node $s$ to (the acyclic graph) $G_\pi$ and for each other node $j$, it adds an arc $(s, j)$ with a cost of 0.  Then for each node $j$ in $G_\pi$, the algorithm determines the shortest path distance $d(j)$ in $G_\pi$ from node $s$ to node $j$; these values can be computed in linear time for an acyclic graph.

\paragraph{Case I:  $d(j) \ge -\sqrt{k}$ for all nodes $j\in I$:}
For each integer $q$ with $-\sqrt{k} \le q \le -1$, let 
$X_q = \{ j\in I :\,  d(j) = -q\}$.  This gives an anti-chain partition of $I$;  thus we have $|X_q|\ge\sqrt{k}$ for the largest one among these sets.  After running \textsc{Cut-Relabel}$(S)$ for the nodes reachable from the largest anti-chain $X_q$, the number of improvable nodes reduces by at least 
$\sqrt{k}$ in $O(m)$ time.   

\paragraph{Case II:  $\min_{j\in I} d(j) < -\sqrt{k}$:}
 In this case, there exists a directed path $P$ in $G_\pi$ that contains improvable arcs $(v_1,w_1), (v_2,w_2), \ldots, (v_t,w_t)$ for $t\ge \sqrt{k}$ in this order.
 We now describe the  subroutine  \textsc{Eliminate-Chain} after which none of the nodes in $w_i$ are improvable, and no new improvable nodes are created.

We start with the original variant of the subroutine used in \textsc{Refine}.
The nodes $w_i$ are processed in reverse order.  For each node $w_i$, $i=t,t-1,\ldots 1$, find the sets $S_i$ of nodes reachable from $w_i$ in $G_\pi$, and run 
 \textsc{Cut-Relabel}$(S_i)$. No new improvable arcs are created, and  if $v\notin S_i$ for any improvable arc $(v,w_i)$, then $i$ is not improvable after the change. It is easy to see that $S_i\subsetneq S_j$ for  all $1\le j<i \le t$. 

 If $v\in S_i$ for an improvable arc $(v,w_i)$ at any iteration, then we discover a negative cost cycle containing $(v,w_i)$. The subroutine \textsc{Refine} terminates at this point by returning this cycle. Goldberg \cite{Goldberg1995} presents an efficient $O(m)$ implementation of \textsc{Refine} by exploiting that the sets $S_i$ are nested.  The implementation (temporarily) contracts the $S_i$ sets, and maintains a data structure using priority queues.

\medskip

We now describe the  variant of \textsc{Eliminate-Chain} used
in  \textsc{Balanced-Refine}. We say that an arc $(v, w)$ is \emph{eligible} if $(v, w)$ is improvable and if $(v, w)$ is not contained in an admissible cycle.   (This definition is not relevant in \textsc{Refine}, since that algorithm terminates if an improvable arc is in an admissible cycle.)   We say that a node $w$ is \emph{eligible} if there is an eligible arc directed into $w$.  Initially, $G_w$ is acyclic, and hence $w_j$ is eligible for all $j \in [t]$.

Now consider the iteration in which the eligible node $w_i$ is selected.  We note that $w_i$ is not eligible after running \textsc{Cut-Relabel}$(S_i)$.  This is because for any arc $(v, w_i )$ that is still improvable after \textsc{Cut-Relabel}$(S_i)$, we must have $v \in S_i$, implying that $(v, w_i)$ was not eligible.   
 
 Let us select the smallest  index $j$ such that after \textsc{Cut-Relabel}$(S_i)$, $w_j$ becomes reachable from $w_i$ in $G_\pi$; that is, $w_j$ enters $S_i$. Let us analyze the case when  $j < i$; note that  $S_i$ also  contains every node on the subpath in $P$ from $w_j$ to $w_i$.  Thus,  $w_\ell$ is not eligible for any $\ell\in [j,i]$ after  \textsc{Cut-Relabel}$(S_i)$.   At the subsequent iteration of  \textsc{Eliminate-Chain} we skip all nodes in $S_i$ and instead select $w_{j-1}$, which is eligible.

After running  \textsc{Cut-Relabel}$(S_i)$, \textsc{Eliminate-Chain} temporarily contracts $S_i$ in the same way as the version used in \textsc{Refine}. Thus,  \textsc{Balanced-Refine} uses essentially the same implementation and data structures as in \cite{Goldberg1995}.

 At the end of \textsc{Eliminate-Chain}, we uncontract all $S_i$'s. Some of the $w_i$'s may now be improvable, however, all improvable arcs incident to them are contained in directed admissible cycles. 
At the next call to \textsc{Decycle}, the algorithm would contract any improvable arc that was not eligible.  Subsequently, the number of improvable nodes will have decreased by at least $\sqrt{k}$.


\section{The shortest path algorithm}\label{sec:shortest-path}
In this section, we assume that a $3$-min-balanced directed graph $G=(N,A,c)$ is given, along with  a component hierarchy $(V\cup N,E,r,a)$ for $G$ (see Definition~\ref{dec:comp-hi}). The algorithm described in this section is an adaptation of Thorup's \cite{Thorup1999} result to the setting of balanced directed graphs. We use the word RAM model throughout this section.

We assume that the  input cost function $c$ is $3$-min-balanced, integral, and strictly positive.
This is justified by Theorem~\ref{thm:balance-main}: in time $O(m\sqrt{n}\log n)$, one can obtain a strictly positive and $1/(4n^3)$-integral reduced cost $c^\pi$ such that $G^\pi=(N,A,c^\pi)$ is 3-min-balanced.  
Since for any $i$--$j$ path $P$, $c^\pi(P)=c(P)-\pi(i)+\pi(j)$, the set of shortest paths between any two nodes is the same in $G$ and $G^\pi$. Integrality can be assumed after multiplying the relabelled cost by $4n^3$ (recall that $n$ is a power of 2); this again does not change the set of  shortest paths.

\subsection{Upper bounds for the component hierarchy}
In the component hierarchy $(V\cup N,E,r,a)$, recall that for a vertex $v\in V$, $\desc(v)\subseteq V\cup N$ denotes the set of descendants of $v$ (with $v\in\desc(v)$). We introduce the shorthand notation $\desc(v,N)=\desc(v)\cap N$ and $\desc(v,V)=\desc(v)\cap V$. 
For a node $u\in V\cup N$, the \emph{height}  $h(u)$ is the length of the longest path between $u$ and a node in $\desc(u)$; in particular, $h(u)=0$ for   $u\in N$.

 We define the functions $\U,\eta:V\to \Q$ recursively, in non-decreasing order of $h(u)$ as follows. 
\begin{equation}\label{eq:U-def}
\begin{aligned}
\U(v)&:=3 a(v)(|\children(v)|-1)+\sum_{v'\in \children(v)\setminus N} \U(v')\, ,\\
\eta(v)&:=\left\lceil \frac{\U(v)}{a(v)}\right\rceil\, .
\end{aligned}
\end{equation}
These values will be relevant for the \emph{buckets} in the algorithm. As shown in the next lemma, $\U(v)$ is a bound on the length of a shortest path between any two nodes in $\desc(v,N)$; we will associate $\eta(v)+1$  buckets with each vertex $v\in V$.
\begin{lemma}\label{lem:U-bound}
Let $(V\cup N,E,r,a)$ be a component hierarchy for a directed graph $G=(N,A,c)$, and let $\U,\eta$ be as in \eqref{eq:U-def}. For any pair of nodes $i,j\in N$    and $v=\lca(i,j)$, there is an $i$--$j$ path $P$ in $\desc(v,N)$ of length at most $\U(v)$.  In addition,  
\[
\sum_{v\in V} \eta(v)< 7|N|\, .
\]
\end{lemma}

\begin{proof}
Let $i,j\in N$ and $v=\lca(i,j)$. The proof is by induction on $h(v)$. 
Consider the $i$--$j$ path $P'$ in $\desc(v)$ such that 
$c(e)\le 3a(v)$ for all $e\in P'$, as guaranteed by the property of the component hierarchy. 

In the base case $h(v)=1$, the bound is immediate, since $P'$ has at most 
$|\children(v)|-1$ arcs. Assume now $h(v)>1$, and that the statement holds  for any $i',j'$ with $h(\lca(i',j'))<h(v)$.  One can choose an $i$--$j$ path $P$ that satisfies the following property for each child $u$ of $v$.   If $i'$ and $j'$ are the first and last nodes of $P$ that are in $\desc(u)$, then the subpath in $P$ from $i'$ to $j'$ consists of nodes of $\desc(u)$.   
By the inductive hypothesis, for each child $u$ of $v$, the length of the subpath in 
$\desc(u)$ is at most $\U(u)$.  There are at most $|\children(v)|-1$ arcs in $P$ between different $\desc(u)$ subpaths; their cost is at most $3 a(v)(|\children(v)|-1)$. 
Thus, the bound $c(P)\le \U(v)$ follows.

Let us now turn to the second statement. 
 We analyze the contribution of each $i\in N$ to the sum $\sum_{v\in V} {\U(v)}/{a(v)}$. Let $i=v_0,v_1,v_2,\ldots,v_k=r$ be the unique path in the tree $(V\cup N,E)$ from $i$ to the root; thus, $p(v_t)=v_{t+1}$ for $t=0,\ldots,k-1$. Then, the contribution of $i$ to each $\U(v_t)$ is less than $3 a(v_1)$. Using that $a(v_{t+1})\ge 2 a(v_t)$ for each $t=0,\ldots,k-1$, we see that 
\[
\sum_{v\in V} \frac{\U(v)}{a(v)}< 3\sum_{i\in N}\sum_{t=1}^\infty \frac{1}{2^{t-1}}< 6|N|\, .
\]
The statement follows noting also that $|V|\le |N|-1$, since $(V\cup N,E)$ is a tree with leaves $N$, and $\eta(v)<1+({\U(v)}/{a(v)})$ for all $v\in V\setminus N$.
\end{proof}

\subsection{Overview of the algorithm}
Given the  input directed graph $G=(N,A,c)$, our goal is to compute the shortest path distances from a \emph{source node} $s\in N$ to all nodes in $N$. We assume that a positive integer cost function and a component hierarchy are given as above. We start with an informal overview and highlight some key ideas of the analysis.

The algorithm is a bucket-based \emph{label setting} algorithm, similarly to a bucket-based implementation of Dijkstra's algorithm. 
For each node $i\in N$, we maintain an upper bound $D(i)$ on the true distance $d(i)$ from $s$, and gradually extend the set $S$ of permanent nodes. Initially, $D(s)=0$ and $D(i)=\infty$ for $i\in N\setminus\{s\}$ and $S=\{s\}$. At the iteration at which $i$ enters $S$, $D(i)=d(i)$ will be guaranteed.

Recall that Dijkstra's algorithm always selects a next node $j$ to enter $S$ with $j\in\arg\min\{D(i):\, i\in N\setminus S\}$. To obtain an $O(m)$ algorithm, we relax this condition, and always add a new node $j\in N\setminus S$ to $S$ such that
\begin{equation}\label{eq:approx-min}
D(j)\le  D(i)+b(i,j)
\quad \forall i\in N\setminus S\, .
\end{equation}
In accordance with this rule, the next lemma formulates the conditions that guarantee the correctness of our algorithm. 
\begin{lemma}\label{lem:shortest-correctness}
Given a directed graph $G=(N,A,c)$ with $c\in\R^N_{\nn}$ and a source node $s\in N$, assume that an algorithm proceeds by adding nodes in $N$ one-by-one to a set $S$ such that the following two invariants are maintained at every iteration:
\begin{enumerate}[(a)]
\item\label{shortest:i} For all $j\in S$ and $i\in N\setminus S$, $D(j)\le D(i)+b(i,j)$.
\item\label{shortest:ii} For all $j\in N\setminus S$, $D(j)$ is the length of a shortest path from $s$ to $j$ inside the node set $S\cup \{j\}$.
\end{enumerate}
Further, assume that initially $D(s)=0$ and $s$ is the first node added to $S$.
Then, at any point of the algorithm, for every $j\in S$, we have $D(j)=d(j)$ and $S$ contains a shortest $s$--$j$ path.
\end{lemma}
\begin{proof} For convenience, suppose that that the nodes are relabelled such that node $i$ is the $i$-th node added to $S$. The lemma is true for node $1=s$, since $D(1)=d(1)=0$. We now assume inductively that the lemma is true for nodes $\ell= 1$ to $i$, and we  prove it for node $i+1$.  

Let $P$ be any path from node $1$ to node $i + 1$.  We  show  $c(P) \ge D(i+1)$; together with \eqref{shortest:ii}, this implies $D(i+1) = d(i+1)$.      

Let $V(P)$ be the vertices of $P$.   If $V(P)\subseteq  \{1,\ldots, i+1\}$, then  
$c(P) \ge D(i+1)$ by \eqref{shortest:ii}.  Otherwise, let $j$ be the first vertex of $P$ that is not in $\{1, \ldots, i+1\}$.  Let $P'$ be the subpath of $P$ from $1$ to $j$.  Then at the iteration in which node $i+1$ is added to $S$, we have
\[
c(P)    \ge c(P') + b(j, i+1)\ge 
 D(j) + b(j, i+1) \ge D(i+1)\, ,
\]
where the second inequality follows by \eqref{shortest:i}. This completes the proof.
\end{proof}

We rely on the component hierarchy and the use of buckets to efficiently implement the selection property \eqref{eq:approx-min}. We will also have (possibly infinite) $D(v)$ values for certain vertices $v\in V$.
Throughout, we maintain a set of \emph{active} vertices (we describe the treatment of active vertices in more detail later).
Initially, the root $r$ is the only active vertex and all other vertices are \emph{inactive}.
At any point, the active vertices form an upper ideal (i.e., all ancestors of an active vertex are also active). Once all their descendants are added to $S$, vertices in $V$ also enter $S$ (become permanent); the algorithm terminates when $r$ is added to $S$. A vertex is active during the iterations from its activation until it is made permanent.
 One of the active vertices will be the \emph{current vertex},  denoted as $\CV$ and initalized as $\CV=r$. This plays a special role: in particular, nodes added to $S$ will always be among the children of $\CV$.

 A vertex $v$ is called a \emph{highest inactive vertex} if $v$ is inactive and $\parent(v)$ is active. For an inactive vertex $v$, we let $\GIA(v)$ denote its \emph{highest inactive ancestor}: $\GIA(v)=v$ if $v$ is a highest inactive vertex; otherwise, $\GIA(v)$ is $v$'s unique ancestor that is a highest inactive vertex.
 
 The next lemma, proved in Section~\ref{sec:shortest-analysis}, shows that for every active vertex $v$, $D(v)$ is a lower bound on $\min\{D(j):\, j\in \desc(v,N)\}$, and when a node in $j\in\desc(v,N)$ is added to $S$, $D(j)$ is within $a(v)$ from $D(v)$. 
 \begin{restatable}{lemma}{activep}\label{lem:active-parent}
Let $j\in N\setminus S$ and let $v$ be an active ancestor of $j$. Then,
 $D(v)\le D(j)$.  In the iteration when $j$ is added to $S$, we also have $D(j)<D(v)+a(v)$.
\end{restatable}
Recalling the property of the component hierarchy that $b(i,j)\ge a(v)$ for $v=\lca(i,j)$, this immediately implies property \eqref{eq:approx-min}.
\paragraph{Buckets}
The choice of $\CV$ and the sequence of nodes added to $S$ is guided by the use of buckets associated with the vertices $v\in V$. The buckets of $v$ are created when $v$ is activated by the \textsc{Activate}$(v)$ subroutine. Before activation, $v$ was a highest inactive vertex, and for all such vertices, we maintain $D(v)=\min\{D(j):\, j\in\desc(v,N)\}$ using the Split/FindMin data structure. 
At activation, $L(v)$ is set to $a(v)\cdot\lfloor D(v)/a(v) \rfloor$.   Then an array  
$\eta(v)+1$ buckets is created for vertex $v$, indexed from $0$ to $\eta(v)$.   The value range of the bucket with index $k$ is  $[L(v) + k a(v),
L(v) + (k+1) a(v))$.  We let $U(v) := L(v) + (\eta(v)+1) a(v)$ denote the upper range of the last bucket for vertex $v$. We place a child $x$ of $v$ in the bucket whose value range contains $D(x)$, or leave it unassigned if 
 $D(x) > U(v)$.

An important feature of the algorithm is that the value range of the buckets at $v$, created at activation, contains the  $d(i)$ values for all $i\in \desc(v,N)$ (Lemma~\ref{lem:desc-upper-bound}). We now highlight the reason behind this. 
At the iteration at which $v$ is activated, let $i = \arg\min\{D(j) :\, j \in \desc(v, N)\}$.   One can show that $d(i) = D(i)$, and that $d(j) \ge d(i)$ for all $j \in \desc(v, N)$. After activation, we have $L(v) \le D(i) \le L(v) + a(v)$.    By Lemma~\ref{lem:U-bound}, for any other node $j \in \desc(v, N)$, there is a path in $G$ with node $i$ to node $j$ of length at most $\Gamma(v)$.  Thus, $d(j) \le d(i) + \Gamma(v) \le L(v) + (\eta(v)+1) a(v) = U(v)$.  

The \emph{current index} $\CI(v)$, initialized as 0, refers to the index of the first nonempty bucket, called the \emph{current bucket}. We will maintain $D(v)$ as the lower endpoint of the current bucket, augmented by $a(v)$ every iteration the current bucket becomes empty. The vertex $v$ is made permanent once $\CI(v)=\eta(v)+1$, that is, all its buckets have been exhausted.

Recall also from Lemma~\ref{lem:U-bound} that the overall number of buckets for all vertices is bounded as $O(n)$; this enables an $O(n)$  running time bound on the  operations involving buckets.

\paragraph{The trajectory of the current vertex}
The algorithm is guided by the movement of the current vertex $\CV$ that explores the component hierarchy. Initially, it moves down from the root $r$ to the source node $s$, activating all vertices along the $r$--$s$ path. 
As long as the current bucket at $\CV$ contains a node, we add such nodes to $S$. Whenever a node $i$ is added to $S$, the subroutine \textsc{Update}$(i)$ scans over the outgoing arcs $(i,j)$, and updates the estimates $D(j)$ to $\min\{D(j),D(i)+c(i,j)\}$ as in Dijkstra's algorithm. This requires some additional updates in the data structure, i.e., moving $j$ to a different bucket if its parent $p(j)$ is active, or updating the $D(w)$ value of its highest inactive ancestor.

If the current bucket $B$ at $v=\CV$ contains some vertices but no nodes, then $\CV$ moves down to a child vertex, and also activates it in case it had not yet been active. If $B$ is empty and if $B$ is not the last bucket of $v$,  then we move the current bucket to the next one, i.e., increment $\CI(v)$ by 1, and increase $D(v)$ by $a(v)$. If the last bucket at $v$ becomes empty, then we make $v$ permanent. At this point, all nodes and vertices in $\desc(v)$ must have been already made permanent. The algorithm then replaces $\CV$ by $p(v)$ if $v\neq r$.  The algorithm terminates once the last bucket at the root $r$ becomes empty and $r$ is made permanent.

After incrementing $\CI(v)$ in the case that $B$ is empty, we proceed to the next bucket with no change in $\CV$ if $v = r$ or if the new $D(v)$ value is less than $D(p(v)) + a(p(v))$.   On the other hand, if $D(v)\ge D(p(v))+a(p(v))$, then the current vertex CV moves up to $p(v)$, and $v$ is moved from the current bucket at $p(v)$ to a higher bucket. Overall, this scheme allows  $D(\CV)$ to be approximately minimal among the labels of active vertices, and thereby enabling the properties asserted in Lemma~\ref{lem:active-parent}.
    
Finally, if the last bucket at $v$ becomes empty, then we make $v$ permanent; at this point, all nodes and vertices in $\desc(v)$ must have been already made permanent. The algorithm terminates once the last bucket at the root $r$ becomes empty and $r$ is made permanent.

\subsection{Description of the algorithm}\label{sec:shortest-describe}

A more formal description of the algorithm with pseudocodes is in order. Recall the basic notation regarding component hierarchies from  Section~\ref{sec:prelim}: $\parent(v)$ (parent of $v$); $\children(v)$ (children of $v$); $\desc(v)$ (descendant of $v$, refined as $\desc(v,N)$ for nodes and $\desc(v,V)$ for vertices); $\lca(u,v)$  (least common ancestor of $v$).

The set $S\subseteq N\cup V$ denotes the set of \emph{permanent} nodes and vertices, initialized as $S=\emptyset$; the first node entering will be the source $s$. Shortest paths will be maintained using predecessor arcs: for each $i\in N\setminus \{s\}$ with $D(i)<\infty$, $\pred(i)\in S$ is an in-neighbour such that $D(i)=D(\pred(i))+c(\pred(i),i)$. The graph of the arcs $(\pred(i),i)$ is acylic, and contains a path from the source $s$ to every node $i\in N$ with $D(i)<\infty$.

The description of the two main subroutines, \textsc{Activate} and \textsc{Update} follows. 

\paragraph{The \textsc{Activate} subroutine and buckets}
Each vertex $v\in V$ can be \emph{active} or \emph{inactive}.
 One of the active vertices will be $\CV$, the current vertex, initalized as $\CV = r$.

 The labels are defined for all nodes (initially as $D(s)=0$ and $D(i)=\infty$ for $i\in N\setminus\{s\}$), for all active vertices, and for all highest inactive vertices. For the latter set, we maintain $D(v) = \min\{D(i) :\, i \in\desc(v)\}$ using the Split/FindMin data structure, as
  detailed in Section~\ref{sec:split-findmin}. For all other inactive vertices, the labels $D(v)$ are undefined.

The \textsc{Activate}$(v)$ subroutine (Algorithm~\ref{alg:activate}) is called the first time $\CV$ is set to $v$. We create an array of $\eta(v)+1$ empty buckets, indexed $k=0,\ldots,\eta(v)$, and denoted as $\bck(v,k)$. The buckets correspond to intervals $[\Lower(v,k),\Upper(v,k))$ of length $a(v)$. The 0th bucket starts at $\LB(v)$, which equals $D(v)$ rounded down to the nearest integer multiple of $a(v)$ (recall this is an integer power of 2).

For $x\in V\cup N$, the \textsc{MoveToBucket}$(x)$  procedure (Algorithm~\ref{alg:movetobucket}) checks if $D(x)$ falls in the value range of a bucket at the parent  $v = p(j)$, places it in such a bucket, and if it was previously in a bucket, deletes it from there.
 
\begin{algorithm}[htb]
    \caption{The \textsc{Activate}  subroutine}\label{alg:activate}
    \begin{algorithmic}[1]
        \Procedure{Activate}{$v$}
         \State $\LB(v)\gets  a(v)\left\lfloor \frac{D(v)}{a(v)}\right\rfloor$\ ;
         \State $D(v)\gets \LB(v)$; $\CI(v)\gets 0$ ;
\For{$k=0,\ldots,\eta(v)$}
 \State $\bck(v,k)\gets \emptyset$ ;
 \State $\Lower(v,k)\gets \LB(v)+k a(v)$ ;
 \State $\Upper(v,k)\gets \LB(v)+(k+1)a(v)$ ;
 \EndFor
 \State $\UB(v)\gets L(v)+(\eta(v)+1)a(v)$ ;
 \For{$w\in \children(v)\cap V$}
                    \State $D(w)\gets\min\{D(i): i\in \desc(v)\}$ ;\Comment{using the Split/FindMin data structure}
                    \State \Call{MoveToBucket}{$w$} ;
                    \EndFor
                      \For{$j\in \children(v)\cap N$}
                    \State \Call{MoveToBucket}{$j$} ;
                    \EndFor
          \EndProcedure
                  \end{algorithmic}

\end{algorithm}

\begin{algorithm}[htb]
    \caption{The  \textsc{MoveToBucket} subroutine}\label{alg:movetobucket}
    \begin{algorithmic}[1]
        \Procedure{MoveToBucket}{$x$}
        \State $v\gets p(x)$ ;
        \If{$v$ is active and $D(x)<\UB(v)$}
         \State $k\gets \left\lfloor\frac{D(x)-\LB(v)}{a(v)}\right\rfloor$ ;
         \If{$x\notin \bck(v,k)$}
         \State delete $x$ from its current  bucket (if any) ;
         \State add $x$ to $\bck(v,k)$ ;
         \EndIf
        \EndIf
              \EndProcedure
        \end{algorithmic}

\end{algorithm}

\paragraph{The \textsc{Update} subroutine}
The \textsc{Update} subroutine (Algorithm~\ref{alg:update}) performs the  label update step once a node $i$ is made permanent, similarly to Dijkstra's algorithm. For every outgoing arc $(i,j)$, if $D(i)+c(i,j)$ is strictly less than the current label $D(j)$, we reduce $D(j)$ to this value, and set the predecessor $\pred(j)$ to $i$. If the parent $p(j)$ is active, we call \textsc{MoveToBucket}{$(j)$} to update the bucket containing $j$. Otherwise, we update $D(w)$ for $w=\GIA(j)$, i.e., the highest inactive ancestor of $j$, using Split/FindMin.

\begin{algorithm}[htb!]
    \caption{The \textsc{Update} subroutine}\label{alg:update}
        \begin{algorithmic}[1]
                 \Procedure{Update}{$i$}
                    \For {$(i,j)\in A(i)$}
                    \If {$D(i)+c(i,j)<D(j)$}
                    \State $D(j)\gets D(i)+c(i,j)$ ; $\pred(j)\gets i$ ;
                    \If {$p(j)$ is active} \Call{MoveToBucket}{$j$} ;
                    \Else \ \
                    $w\gets \GIA(j)$ ; $D(w)\gets \min\{D(j), D(w)\}$ ; \\ \Comment{using the Split/FindMin data structure}
                    \EndIf
                    \EndIf
                    \EndFor

          \EndProcedure
        \end{algorithmic}
\end{algorithm}

\paragraph{The overall algorithm}
The overall algorithm is shown in Algorithm~\ref{alg:shortest}. 
Initially, the current vertex is set as the root: $\CV=r$. At any given iteration, we let $v=\CV$ and let $B$ denote the current bucket at $v$, i.e., $B=\bck(w,\CI(v))$.

If $B$ contains a node $i\in N$, we make it permanent, i.e., add it to $S$, and call   \textsc{Update}$(i)$ to update the labels for each out-neighbour $j$ of $i$. If $B$ contains no nodes but some vertices, we move $\CV$ to such a vertex $w$, and activate it if necessary.

The remaining possibility  is when  the bucket $B$ becomes empty in the current iteration.  We increment the counter $\CI(v)$ by 1 and accordingly update $D(v)$ to $D(v)+a(v)$, the starting point of the new current bucket. In case $\CI(v)=\eta(v)+1$, i.e., if $B$ was already the final bucket, then we make $v$ permanent, and unless $v=r$, we move $\CV$ up to the parent $p(v)$. If $v=r$ then the algorithm terminates.

Otherwise, if $\CI(v)\le\eta(v)$, we check if the updated value $D(v)\ge D(p(v))+a(p(v))$, i.e., if the update requires moving $v$ to a higher bucket at $p(v)$ (assuming $v\neq r$). If this is the case, $\CV$  moves up to $p(v)$; otherwise, we proceed with $\CV=v$.

\begin{algorithm}
    \caption{\textsc{Shortest-Paths}}\label{alg:shortest}
    \begin{algorithmic}[1]
        \Require{A directed graph $G=(N,A,c)$  with  $c\in\Z_{\po}^A$, source node $s\in N$, a component hierarchy $(V\cup N,E,a)$ for $G$.}
                \Ensure{Shortest path labels for each $i\in N$ from $s$.}
                \State $S\gets \emptyset$ ;
                \State $D(s)\gets 0$ ; $D(r)\gets 0$ ;
                \For{$j\in N\setminus \{s\}$} $D(j)\gets\infty$ ;\EndFor 
                \For{$v\in V$} compute $\U(v)$ and $\eta(v)$ as in \eqref{eq:U-def} ;\EndFor
                \State $\CV\gets r$ ; \Call{Activate}{$r$} ;
                \While{$r\notin S$}
             \State $v\gets \CV$ ; $B\gets \bck(v,\CI(v))$ ;
             \If {$B\cap N\neq \emptyset$} 
              \State select a node $i\in B\cap N$ and delete $i$ from $B$ ;
              \State $S\gets S\cup \{i\}$ ;
              \State \Call{Update}{$i$} ;
              \ElsIf{$B\cap V\neq\emptyset$}
              \State select a vertex $w\in B\cap V$ ;
              \State $\CV\gets w$ ;
              \If{$w$ is inactive} \Call{Activate}{$w$} ;
              \EndIf
              \Else \Comment{$B=\emptyset$}
                   \State $\CI(v)\gets \CI(v)+1$ ; $D(v)\gets D(v)+a(v)$ ;
                \If{$\CI(v)=\eta(v)+1$}
                    \State $S\gets S\cup\{v\}$ ;
                    \If{$v\neq r$} $\CV\gets p(v)$ ;\EndIf 
                \ElsIf{$v\neq r$ and $D(v)\ge D(p(v))+a(p(v))$}
                    \State $\CV\gets p(v)$ ;
                    \State \Call{MoveToBucket}{$v$} ;
            \EndIf
             \EndIf

\EndWhile
    \State \Return{labels $D(i)$: $i\in N$.} 
        \end{algorithmic}
\end{algorithm}

\subsection{Analysis}\label{sec:shortest-analysis}

\begin{theorem}\label{thm:shortest-main} 
 Algorithm~\ref{alg:shortest} computes shortest paths from node $s\in N$ to all other nodes in  $O(m)$.
\end{theorem}

We prove the theorem in two parts. Lemma~\ref{lem:running} shows the running time bound $O(m)$.
 Correctness follows using Lemma~\ref{lem:shortest-correctness} and Lemma~\ref{lem:active-parent} stated above. To prove the latter lemma, we need one more auxiliary statement (Lemma~\ref{lem:v-w}) that relates the label of an active vertex to that of its active descendants.

\begin{lemma}\label{lem:running}
The total running time of Algorithm~\ref{alg:shortest} is bounded as  $O(m)$.
\end{lemma}
\begin{proof}
The time for initialization is $O(n)$. Let us show that the main \emph{while} cycle is called $O(n)$ times. We consider the cases for $v=\CV$ and current bucket $B$ as {\em (i)} $B$ contains a node, or {\em (ii)} $B$ contains a vertex but no node, or {\em (iii)} $B$ is empty.  

Whenever case {\em (i)} occurs, a node is added to $S$, giving a bound of $O(n)$ for this case.   In case {\em (iii)}, $\CI(v)$ is incremented, and $\CV$ is possibly moved to $p(v)$.   The number of times this can occur is equal to the total number of buckets, which is $O(n)$ by Lemma~\ref{lem:U-bound}.

Let us now turn to case {\em (ii)}. Let $\tau$ denote the distance of the current vertex $\CV$ from the root $r$ in the component hierarchy. Both in the first and the final iteration, $\CV=r$, and thus $\tau=0$. Whenever case {\em (ii)} occurs, $\tau$ increases by one. The only way $\tau$ can decrease is if $\CV$ is moved from a vertex to its parent in case {\em (iii)}.
Thus, the total number of occurrences of case {\em (ii)}  is equal to the total number of  increases in $\tau$, which equals the total number of decreases, in turn bounded by $O(n)$. Thus, each of the three cases can only occur $O(n)$ times, bounding the number of iterations of the \emph{while} cycle.

\medskip

The subroutine \textsc{Update}$(i)$ is called once for each $i\in N$.  At each call,  the arcs in $A(i)$ are scanned.   The time to update $D(j)$ for $(i, j) \in A(i)$ is $O(1)$.   If $p(j)$ is active, then the time to put node $j$ in the correct bucket at $p(j)$ is $O(1)$.  A potential bottleneck occurs when $p(j)$ is inactive and $D(j)$ is updated.  In this case, the algorithm determines $w = \GIA(j)$ and then updates $D(w)$.  
The amortized time to 
 determine $w$ and update $D(w)$ is $O(1)$ using Thorup’s \cite{Thorup1999}  implementation of the Split/FindMin data structure (see Section~\ref{sec:split-findmin}). Thus, the total time of the updates is $O(m)$.

We now consider \textsc{Activate}$(v)$, which is called $O(n)$ times. The total number of buckets is $O(n)$, and 
each $x\in\children(v)$ has to be placed in a bucket; note that $\sum_{v\in V}|\children(v)|\le 2n-1$. The overall time for creating buckets and placing the children in buckets takes $O(n)$. Further, we need to update 
 $D(w)$ for $w \in \children(v)\cap V$.  For each $w$, this is again accomplished using the Split/FindMin data structure in amortized time $O(1)$.

  The total running time of the Split/FindMin operations can be bounded as $O(m)$.
The $O(n)$ bound on the \emph{while} iterations, the total  $O(n)$ on  \textsc{Activate} and $O(m)$ on \textsc{Update} yields the overall $O(m)$ bound.
\end{proof}

The next lemma will be key in proving Lemma~\ref{lem:active-parent}.
\begin{lemma}\label{lem:v-w}  Let $v$ and $w$ be active vertices such that $w\in\desc(v,V)$. Then
\begin{enumerate}[(i)]
    \item \label{v-w:i}
$D(w)\le D(v) + a(v)$; and
\item \label{v-w:ii} if 
 $\CV\in\desc(w,V)$, then  $D(w)+a(w) \le D(v) + a(v)$.
\end{enumerate}
\end{lemma}
\begin{proof}
We start by showing part \eqref{v-w:ii}.   
We prove it for the case $v = p(w)$; this immediately implies the general case.
 We first consider the case that $w$ has just become the current vertex and $v$ was previously the current vertex.   Since $w$ was selected from the current bucket of $v$, it follows that $D(w) < D(v) + a(v)$.   Moreover, $D(w)$, $D(v)$ and $a(v)$ are all integer multiples of $a(w)$.   (In the case that $w$ was just activated, its label $D(w)$ was obtained by rounding its previous label down to the nearest multiple of $a(w)$.) The claim that $D(w) + a(w) \le D(v) + a(v)$ follows.

If $w$ is the current vertex, then $D(w)$ may only change if the current bucket at $w$ is empty, in which case  $D(w)$ is incremented to $D'(w)=D(w)+a(w)$. If $D'(w)\ge D(v)+a(v)$, then the current vertex moves up to $v$, at which point $\CV\notin\desc(w,V)$. Otherwise, $D'(w)< D(v)+a(v)$, implying $D'(w)+a(w)\le D(w)+a(w)$ as above.

In all other iterations when $\CV\in\desc(w,V)$, neither $D(v)$ nor $D(w)$ may change, and therefore the statement remains valid. This completes the proof of part \eqref{v-w:ii}.

Let us now show part \eqref{v-w:i}; we do not assume $v=p(w)$ for this proof. In light of part \eqref{v-w:ii}, we can focus on iterations when $\CV\notin \desc(w,V)$. When $w$ is activated, $w=\CV$. Consider any iteration when $\CV$ leaves $\desc(w,V)$; this means that $\CV$ moves from $w$ to $p(w)$. This happens when the current bucket at $w$ is empty and $D(w)\ge D(p(w))+a(p(w))$; but again using divisibility this means $D(w)= D(p(w))+a(p(w))$. By part \eqref{v-w:ii} applied to $p(w)$ and $v$, it follows that $D(w)= D(p(w))+a(p(w))\le D(v)+a(v)$. In all subsequent iterations until $w$ becomes the current vertex again, $D(w)$ remains unchanged, and $D(v)$ may only increase. Thus, $D(w)\le D(v)+a(v)$ is maintained, implying \eqref{v-w:i}.
\end{proof}
We are ready to show Lemma~\ref{lem:active-parent}, restated here.
\activep*
\begin{proof}
Let us start with the second statement. When $j$ is added to $S$, then $w=p(j)$ must be the current vertex, and $j$ is in the current bucket at $v$, that is, $D(w)\le D(j)< D(w)+a(w)$. According to Lemma~\ref{lem:v-w}\eqref{v-w:ii}, we have $D(w)+a(w)\le D(v)+a(v)$. Thus, the second statement holds.

We now prove the first statement by induction on the number of iterations. The statement clearly holds at initialization: $r$ is the only active vertex.  $D(r)=0$, $D(s)=0$, and $D(i)=\infty$ for $i\in N\setminus\{s\}$. Assume $D(v)\le D(j)$ holds at the beginning of the current iteration for every pair $j$ and $v$ such that $j\in N$, $v\in V$ is active, and $j\in\desc(v,N)$. The label of an active vertex may only increase, and the label of a node may only decrease in the algorithm; we analyze the two cases separately.

Consider a pair of $v$ and $j$ such that $D(v)$ increases. $D(v)$ may only change when $\CV=v$, and the current bucket at $v$ is empty; the new value is set to $D'(v)=D(v)+a(v)$. We claim that  $D'(v)\le D(j)$ holds. If $j\in\children(v)$, then this is true because the current bucket was empty.  Otherwise, let $w\in\children(v)\cap V$ be the vertex following $v$ on the $v$--$j$ path in the component hierarchy. Since the first bucket is empty, we must have $D(v)+a(v)\le D(w)$. By induction, we have $D(w)\le D(j)$; thus, $D'(v)\le D(j)$ must still hold.

Consider now a pair $v$ and $j$ such that $D(j)$ decreases. This can happen in a call to \textsc{Update}$(i)$ such that $(i,j)\in E$ and $D'(j)=D(i)+c(i,j)<D(j)$. We need to show $D'(j)\ge D(v)$.

Let $z=\lca(i,j)$; by the property of the component hierarchy, we have $c(i,j)\ge a(z)$. By induction, $D(z)\le D(i)$, and thus $D(z)+a(z)\le D(i)+ c(i,j)$. 
Since $z$ and $v$ are both on the path from $j$ to $r$, either $z \in \desc(v, V)$ or $v \in \desc(z, V)$.

If $v\in \desc(z,V)$, then $D(v)\le D(z)+a(z)\le D(i)+c(i,j)=D'(j)$ using  Lemma~\ref{lem:v-w}\eqref{v-w:i}.
If $z\in\desc(v,V)$, then also $i\in\desc(v,V)$, and thus $D(v)\le D(i)<D(i)+c(i,j)=D'(j)$ by induction. This completes the proof of the first statement.
\end{proof}

\begin{lemma}\label{lem:desc-upper-bound}
For every vertex $v\in V$ and descendant $i\in\desc(v)$, $d(i)< \UB(v)$.
\end{lemma}
\begin{proof}
Let $D(.)$ denote the labels immediately prior to the activation of vertex $v$, and let $D'(.)$ be the labels immediately after activation.  Then 
$D'(v) \le D(v) < D'(v) + a(v)$. 
Let $i := \arg\min\{D(j)\, :\, j \in \desc(v, N)\}$.   Then $d(i) \le D(i) = D(v)$.   By Lemma~\ref{lem:U-bound}, for all $j \in \desc(v, N)$, 
\[
\begin{aligned}
d(j) &\le d(i) + \Gamma(v)  \le D(v) + \Gamma(v) 
\le D(v) + \eta(v) a(v) \\
&\le D'(v) + (\eta(v) + 1)a(v) - 1  
= U(v) - 1\, .\end{aligned}\]  
\end{proof}

We are ready to prove Theorem~\ref{thm:shortest-main}.

\begin{proof}[Proof of Theorem~\ref{thm:shortest-main}]
Lemma~\ref{lem:running} provides the running time analysis. It remains to show that $D(i)=d(i)$ for every $i\in S$, and that the algorithm terminates with $N\subseteq S$.

We can use Lemma~\ref{lem:shortest-correctness} to show that the algorithm correctly sets the labels inside $S$. For this, we need to verify the following two properties:
\begin{enumerate}[(a)]
\item For all $j\in S$ and $i\in N\setminus S$, $D(j)\le D(i)+b(i,j)$.
\item For all $j\in N\setminus S$, $D(j)$ is the length shortest path from $s$ to $j$ inside the node set $S\cup \{j\}$.
\end{enumerate}
The proof of  \eqref{shortest:ii} follows the same argument as for Dijkstra's algorithm, see e.g. \cite[Section 4.5]{AMO}.
Part \eqref{shortest:i}  clearly holds in the first step when $S=\{r\}$. At the iteration when a node $j$ is added to $S$, consider any $i\in N\setminus S$, $i\neq j$, and let $v=\lca(i,j)$. Then, Lemma~\ref{lem:v-w} shows $D(j)-a(v)\le D(v)\le D(i)$. The claim follows since $b(i,j)\ge a(v)$ is a property of the component hierarchy.

It remains to show that $N\subseteq S$ at termination, i.e., at the iteration that sets $D(r)=\UB(r)$ and makes $r$ permanent. For a contradiction, let $j\in N\setminus S$ at the this iteration.
 Let $P = i_1, i_2, \ldots, i_k$ (where $i_1 = s$ and 
$i_k = j$)  be a shortest path from node $s$ to node $j$. Clearly, $k\ge 2$, and without loss of generality, let us assume that each node $i_t$, $t\le k-1$ was added to $S$ during the algorithm (or else we can replace $j$ by the first node $i_t$ of $P$ not added to $S$).

In the iteration when  $h=i_{k-1}$ was added to $S$, we had $D(h)=d(h)$ as shown above. Further, \textsc{Update}$(i)$ updated $D(j)$ to $D(h)+c_{hj}=c(P)=d(h)$. Clearly, $D(h)=d(h)$ for the rest of the algorithm. According to Lemma~\ref{lem:active-parent}, the final iteration has 
\[
\UB(r)=D(r)\le D(h)=d(h)<\UB(r)\, ,
\]
where the first equality follows by the termination condition, and the last inequality by  Lemma~\ref{lem:desc-upper-bound}. This completes the proof.
\end{proof}

\subsection{The Split/FindMin data structure}\label{sec:split-findmin}

For each highest inactive vertex $v$, the algorithm needs to be able to compute $D(v) = \min\{D(j) :\, j \in \desc(v, N)\}$.   To accomplish this, we will use the Split/FindMin data structure.   Before reviewing this data structure, we note that in addition to computing $D(v)$ for highest inactive vertices, the data structure will need to be updated whenever either of the following algorithmic operations takes place:
\begin{itemize}
\item  When a highest inactive vertex is activated, the subset $\children(v) \cap V$ all become highest inactive vertices.
\item  In step \textsc{Update}$(i)$, if $D(j)$ is updated, then $D(v)$ should be updated for $v = \GIA(j)$. 
\end{itemize}

The steps can be implemented using the Split/FindMin data structure. This was first introduced by Gabow \cite{Gabow85} for the maximum weight matching problem, and can be stated as follows (see also \cite{Pettie2014}).  The data structure is initialized with a sequence $E = \{e_1, \ldots, e_n\}$ of $n$ weighted elements.  
At each iteration, there is a set $\mathcal{S}$, which is a partition of $E$ into consecutive subsequences.  For every element $e_i$, we maintain a key value $\kappa(e_i)$.   At a given operation described below, we  
let $\mathcal{S}(e_i)$ denote the unique subsequence $\mathcal{S}$ that contains $e_i$.  Note that $\mathcal{S}$ and 
$\mathcal{S}(e_i)$ are modified whenever a split operation is called.

The operations are as follows:
\begin{itemize}
 \item init$(e_1,e_2,\ldots,e_n)$: Create a sequence set $\mathcal{S}\gets\{(e_1,e_2,\ldots,e_n)\}$ with $\kappa(e_i)=\infty$ for all $i\in [n]$.
 \item split$(e_i)$: For $\mathcal{S}(e_i)=(e_j,\ldots,e_{i-1},e_i,\ldots, e_k)$,  let $\mathcal{S}\gets (\mathcal{S}\setminus {\mathcal{S}}(e_i))\cup\{(e_j,\ldots,e_{i-1}),(e_i,\ldots, e_k)\}$.
 \item findmin$(e_i)$: Return $\min\{\kappa(e_j): e_j\in \mathcal{S}(e_i)\}$.
 \item descreasekey$(e_i,w)$: Set $\kappa(e_i)\gets \min\{\kappa(e_i),w\}$.
\end{itemize}

To use this data structure for our setting, we take the component hierarchy 
$(V\cup N,E,r,a)$, and impose an arbitrary ordering 
on the children of every vertex $v\in V$. 
This induces a total ordering on the set of leaves $N$; 
we index the node set $N=\{e_1,e_2,\ldots, e_n \}$ accordingly. 
Then, all sets $\desc(v)$ will correspond to contiguous subsequences of nodes. 
Initially, there are no active vertices and $E$ is the set of nodes.  Then $r$ is activated.  In general, when a vertex $v$ is activated, it corresponds to
 performing $|\children(v)| - 1$ splits on the nodes in $\desc(v)$, resulting in a consecutive subsequence for each child of $v$. 
(The nodes in $\children(v)$ correspond to subsequences of length 1.)  
Whenever $D(j)$ is updated, this corresponds to a decreasekey operation.  Using the Split/FindMin data structure for the shortest path algorithm requires at most $n$ findmin operations, at most $n-1$ splits, and at most $m$ decreasekey operations.

For $O(n)$ split and $O(m)$ decreasekey operations with $m\ge n$, Gabow \cite{Gabow85} gave an implementation in $O(m\alpha(m,n))$ total time in the comparison-addition model. This was improved by Thorup to $O(m)$ in the word RAM model, using the atomic heaps data structure by Fredman and Willard \cite{Fredman1994}. The original implementation of fusion trees permits all bitwise operations as well as multiplication.   In a subsequent paper \cite{ThorupAC}, Thorup showed how to implement fusion trees on a mild extension of the $\mathrm{AC}^0$ model, thus avoiding the need for multiplication except for multiplication by powers of 2.

We note that the data Split/FindMin structure was also used in all subsequent papers on shortest path problems using the hierarchy approach \cite{Hagerup2000,Pettie2002,Pettie2004,PR2005}. In the comparison-addition model, an improved bound $O(m\log\alpha(n,m))$ was given by Pettie \cite{Pettie2014}.


\section{Conclusions}\label{sec:conclusions}
In this paper, we have given an $O(mn)$ algorithm for the directed all pairs shortest paths problem with nonnegative integer weights. Our algorithm first replaces the cost function by a reduced cost satisfying an approximate balancing property in $O(m\sqrt{n}\log n)$ time. Subsequently, every shortest path computation can be done in linear time, by adapting Thorup's algorithm \cite{Thorup1999}.

One might wonder if our technique may also lead to an improvement for APSP in the comparison-addition model, where the best running time is $O(mn+n^2\log\log n)$ by Pettie \cite{Pettie2004}. This running time bound is based on multiple bottlenecks. However, as explained in Section~\ref{sec:shortest-overview}, the approximate cost balancing is able to get around the sorting bottleneck of \cite{Pettie2004}. Using the $O(m\log\alpha(n,m))$ implementation of Split/FindMin, an overall $O(mn\log\alpha(n,m))$ might be achievable. 

However, there is one remaining important bottleneck where our algorithm crucially relies on bit-shift operations: the operation \textsc{MoveToBucket}$(j)$, which places a node/vertex in the bucket at $v=p(j)$ containing the value $D(j)$.
Pettie and Ramachandran \cite{PR2005} show that these operations can be efficiently carried out in $O(1)$ amortized time per operation in a bucket-heap data structure, assuming the hierarchy satisfies certain \emph{`balancedness'} property. Section 5 of the paper shows how the `coarse hierarchy' obtainable from a minimum spanning tree and used by Thorup can be transformed to a `balanced hierarchy'. This method does not seem to easily apply to the directed hierarchy concept used in this paper.

Our approximate min-balancing algorithm may be of interest on its own, and has strong connections to the matrix balancing literature as detailed in Section~\ref{sec:balance-overview}. For finding an $(1+\varepsilon)$-min-balanced reduced cost for 
$\varepsilon=O(1)$, our algorithm takes $O\left(\frac1\varepsilon {m\sqrt{n}\log n}\right)$ time. One might wonder if there is an algorithm with the same polynomial term $\tilde O(m\sqrt{n})$ but with a dependence on $\log(1/\varepsilon)$. We note that the algorithm in \cite{Schulman2017} for approximate max-balancing has a $\log(1/\varepsilon)$ dependence.

\paragraph{Acknowledgement} The authors are very grateful to an anonymous referee. Their insightful comments lead to simplifications in some arguments and significant improvements in the presentation.

\bibliographystyle{abbrv}
\bibliography{references}

\end{document}